\documentclass{toc}

\tocdetails{%
  title = {Separating k-Player from t-Player One-Way Communication, with Applications to Data Streams},
  number_of_pages = {35},%%
  number_of_bibitems = {24},%%
  number_of_figures = {0},%%
  conference_version = {ICALP19},%% examples: {EC17},  {FSTTCS18}, {NONE}
  author = {Elbert Du, Michael Mitzenmacher, David Woodruff, and Guang Yang},%%
  authorlist = {Elbert Du, Michael Mitzenmacher, David Woodruff, Guang Yang},%%
  acmclassification = {F.1.3, F.2.3, F.2.1},
  amsclassification = {68Q17, 68W25, 68W20},
  keywords = {streaming, space complexity, hamming norm, approximation, algorithms with predictions, direct sum}
}   %% END AUTHOR-FILLED METADATA

\tocdetails{%
}   %% END tocdetails

\usepackage{cleveref}
\usepackage[shortlabels]{enumitem}
\newcommand{\set}[1]{{\left \{{#1} \right \}}}
\newcommand{\norm}[1]{{\left\| {#1} \right \|}}
\newcommand{\intersect}{\ensuremath{\cap}}
\newcommand{\union}{\ensuremath{\cup}}
\newcommand{\st}{\;|\;}

\newcommand{\cond}{\;|\;}
\newcommand{\Cond}{\;\Big|\;}
\newcommand{\follows}{\ensuremath{\sim}}

\newcommand{\ceil}[1]{ {\left\lceil{#1} \right\rceil}}

\newcommand{\poly}{\mathsf{poly}}
\newcommand{\polylog}{\ensuremath{\mathsf{polylog}}}
\newcommand{\zo}{\ensuremath{ \{ 0, 1 \} }}
\newcommand{\zon}{\ensuremath{ \{ 0, 1 \}^n }}
\newcommand{\error}{\ensuremath{\delta}}
\newcommand{\gap}{\ensuremath{\varepsilon}}

\newcommand{\bigo}[1]{\ensuremath{{{O} \left( {#1} \right)}}}
\newcommand{\bigomega}[1]{\ensuremath{{\Omega \left( {#1} \right)}}}
\newcommand{\littleo}[1]{\ensuremath{{o \left( {#1} \right)}}}

\newcommand{\bigtheta}[1]{\ensuremath{{\Theta\left({#1} \right)}}}

\newcommand{\eqdef}{\ensuremath{:=}}

\newcommand{\E}{\ensuremath{\mathbf{E}}}

\newcommand{\entropy}{\ensuremath{\mathsf{H}}}
\newcommand{\Hmin}{\ensuremath{\mathsf{H}_{\infty}}}

\newcommand{\event}{\ensuremath{\mathcal{E}}}

\newcommand{\remove}[1]{}

\newcommand{\alphabet}{\ensuremath{\Sigma}}
\newcommand{\range}{\ensuremath{\Gamma}}

\newcommand{\dcc}{\ensuremath{\mathbf{DCC}}}
\newcommand{\rcc}{\ensuremath{\mathbf{RCC}}}
\newcommand{\dsc}{\ensuremath{\mathbf{DSC}}}
\newcommand{\rsc}{\ensuremath{\mathbf{RSC}}}
\newcommand{\rscT}{\ensuremath{\mathbf{RSC}^T}}
\newcommand{\ic}{\ensuremath{\mathbf{IC}}}
\newcommand{\dccone}{\ensuremath{\overrightarrow{\mathbf{DCC}}}}
\newcommand{\rccone}{\ensuremath{\overrightarrow{\mathbf{RCC}}}}
\newcommand{\dscone}{\ensuremath{\overrightarrow{\mathbf{DSC}}}}
\newcommand{\rscone}{\ensuremath{\overrightarrow{\mathbf{RSC}}}}
\newcommand{\rccsim}{\ensuremath{\mathbf{RCC}^{sim}}}
\newcommand{\rcclin}{\ensuremath{\mathbf{RCC}^{LIN}}}
\newcommand{\rcclinT}{\ensuremath{\mathbf{RCC}^{LIN,T}}}
\newcommand{\D}{\ensuremath{\mathbf{D}}}

\newcommand{\DlinT}{\ensuremath{\mathbf{D}^{LIN,T}}}

\newcommand{\sumequal}{\textsc{Sum-Equal} }

\newcommand{\hest}{ \textsc{T} }

\newcommand{\go}{\textsc{Gap-Ort}}
\newcommand{\gose}{\textsc{Gap-Ort-Sum-Equal}}
\newcommand{\auggo}{\textsc{Aug-Index-GOSE}}

\begin{document}

\begin{frontmatter}%%[classification=text] << EDITOR.

%% EDITOR: If abstract fits entirely on first page, you may consider
%% the "classification=text" option, which typesets classifications
%% (as text) directly after the abstract--a preferable arrangement.

%%% !!! AUTHOR Title goes here
%\title{}  %% if conf version exists, see below

%%% !!! If a conference version exists, use instead the following
%%%   (after replacing the conference name)

\title{Separating k-Player from t-Player One-Way Communication, with Applications to Data Streams\titlefootnote{A preliminary version of this paper appeared in the \href{https://drops.dagstuhl.de/opus/volltexte/2019/10673/pdf/LIPIcs-ICALP-2019-97.pdf}{Proceedings of the 46th International Colloquium on Automata, Languages and Programming, 2019}.}}

%%% !!! Replace XXX by one of the following phrases: 
%%%   An extended abstract (if the current version adds or significantly 
%%%         expands the proofs of the main results stated in the conference
%%%         version but most of the main results of the current paper have 
%%%         already been (essentially) stated in the conference version);
%%%   A preliminary version (if the current version contains significant
%%%         new results or significant improvements over the results 
%%%         stated in the conference version); 
%%%   A conference version (in all other cases).
%%% Appropriately modify this text if the paper descends from more than one
%%% conference paper.   Make sure to include the conference version(s)
%%% in the .bib file.

%%% !!! AUTHOR List all authors. In brackets include the author's
%%% **last name** in lower case with no special characters; this
%%% will be used as the unique tag (author ID) to associate
%%% the author with the correct bio sketch at the end of the paper.

%%% List grant acknowledgements in \thanks.
\author[du]{Elbert Du}
\author[mitzenmacher]{Michael Mitzenmacher\thanks{Supported in part by NSF grants CCF-2101140, CNS-2107078, and DMS-2023528, and by a gift to the Center for Research on Computation and Society at Harvard University.}}
\author[woodruff]{David Woodruff}
\author[yang]{Guang Yang}

%%% !!! AUTHOR Dedication, if desired, goes here  
%% \begin{dedication}
%% !!!
%% \end{dedication}

%%%  !!! AUTHOR Abstract goes here
%%%  limit your Abstract to 1920 characters to satisfy the arXiv standard
%%%  no \cite{...} commands in Abstract; citation format in abstract:
%%%   (Jones and Kumar, STOC'14)
%%%   if more than two authors: (Jones et al., STOC'14)
%%%   if journal: (Jones et al., SICOMP 2014)
\begin{abstract}
In a $k$-party communication problem, the $k$ players with inputs $x_1, x_2, \ldots, x_k$, respectively, want to evaluate a function $f(x_1, x_2, \ldots, x_k)$ using as little communication as possible. 
We consider the message-passing model, in which the inputs are partitioned in an arbitrary, possibly worst-case manner, among a smaller number $t$ of players ($t<k$). 
The $t$-player communication cost of computing $f$ can only be smaller than the $k$-player communication cost, since the $t$ players can trivially simulate the $k$-player protocol. 
But how much smaller can it be? 
We study deterministic and randomized protocols in the one-way model, and provide separations for product input distributions, which are optimal for low error probability protocols. We also provide much stronger separations when the input distribution is non-product.  

A key application of our results is in proving lower bounds for data stream algorithms. In particular, we give an optimal $\Omega(\gap^{-2}\log(N) \log \log(mM))$ bits of space lower bound for the fundamental problem of $(1\pm\gap)$-approximating the number $\|x\|_0$ of non-zero entries of an $n$-dimensional vector $x$ after $m$ integer updates each of magnitude at most $M$, and with success probability $\ge 2/3$, in a strict turnstile stream. We additionally prove the matching $\Omega(\gap^{-2}\log(N) \log \log(T))$ space lower bound for the problem when we have access to a heavy hitters oracle with threshold $T$.
Our results match the best known upper bounds when $\gap\ge 1/\polylog(mM)$ and when $T = 2^{\poly(1/\epsilon)}$ respectively.
It also improves on the prior $\Omega(\gap^{-2}\log(mM) )$ lower bound 
and separates the complexity of approximating $L_0$ from approximating the $p$-norm $L_p$ for $p$ bounded away from $0$, since the latter has an $O(\gap^{-2}\log (mM))$ bit upper bound. 
\end{abstract}

\iffalse % DON'T TOUCH THIS LINE
%%%  AUTHOR: Choose the arXiv category that best fits your article.
%%%  A complete list of computer science categories is available here:
%%%     http://arxiv.org/archive/cs
%%%  and in math
%%%     http://arxiv.org/archive/math
%%%  Common categories include:
%%%     cs.CC - Computational Complexity
%%%     cs.CR - Cryptography and Security 
%%%     cs.DS - Data Structures and Algorithms
%%%     cs.LG - Learning
%%%     cs.IT - Information Theory 
%%%     cs.DM - Discrete Mathematics 
%%%     quant-ph - Quantum Physics
%%%     math.PR - Probability
%%%     math.CO - Combinatorics
%%%  You must include at least one category. If you include more than one category, 
%%%  the first one will serve as your main category, and the others will be used for
%%%  crosslisting. 
%%%
%%%  Example:
%%%  \tocarxivcategory{cs.CC,quant-ph}

\tocarxivcategory{cs.CC}

\fi % DON'T TOUCH THIS LINE

\end{frontmatter}

%%%
%%% !!! AUTHOR
%%% Paper goes here...

\section{Introduction}
\label{sec:intro}

Consider a $k$-party communication problem, in which the players have inputs $x_1, x_2, \ldots, x_k$ respectively, and want to compute a function $f(x_1, x_2, \ldots, x_k)$ of their inputs using as little communication as possible. 
We consider the message-passing model, in which the inputs are partitioned in an arbitrary, possibly worst-case manner among a smaller number $t$ of players. 
That is, we partition $\set{1,2,\dots,k}$ into $t$ subsets $S_1, S_2, \ldots, S_t$ such that $\union_{i=1}^t S_i = \set{1,2,\dots,k}$ and $S_i\intersect S_j = \emptyset$ for every $1\le i< j\le t$,
and let the $i$-th player $P_i$ hold the sequence of inputs $y_i\eqdef \left(x_{i_1}, x_{i_2}, \ldots, x_{i_{|S_i|}}\right)$. 
We are still interested in computing the original function $f$.
The total communication required must be smaller than in the original $k$-player setting, since the $t$ players can simulate the protocol involving the original $k$ players. 
A natural question is: {\it how much smaller can the communication be?}

There are many communication models that are possible, but our main motivation for looking at this question comes from applications to data streams, see below, and so we are primarily interested in the {\it one-way number-in-hand} model. In this model, each of the $t$ players can only see its own input. 
The first player composes a message $m_1$ based on its input $y_1$ and sends $m_1$ to the second player. 
The second player takes $m_1$ and its input $y_2$ to compute a message $m_2$ for the third player, and so on. 
The $t$-th (also the last) player, upon receiving the message $m_{t-1}$ from the $(t-1)$-st player, computes the output of the protocol based on $m_{t-1}$ and its own input $y_t$. We sometimes abuse notation and refer to the output as $m_t$. 
The total communication cost is the maximum of $\sum_{i=1}^{t} |m_i|$, where $|m_i|$ denotes the length of the $i$-th message and the maximum is taken over all possible inputs $y_1,\ldots, y_t$ (which is a partition of $\set{x_1, \ldots, x_k}$) and all random coin tosses of the players. 
For streaming applications we are especially interested in $\max_{i\in\set{1, \ldots, t}} |m_i|$. 

To explain the connection to data streams, almost all known lower bound arguments on the memory required of a data stream algorithm are proven via communication complexity, or at least can be reformulated using communication complexity. 
The basic idea is to partition the elements of an input stream contiguously, consisting of say $k$ elements, into a possibly smaller number $t$ of players. 
Then one argues that if there is a data stream algorithm  solving the problem, then the communication problem can be solved by passing the memory contents as messages from player to player. 
Note that this naturally gives rise to the one-way number-in-hand model. 
Since the total communication cost is $t \cdot S$, where $S$ is the size of the memory of the streaming algorithm, if the randomized $t$-player communication complexity of the function $f$ is $CC_t$, we must have $S \geq CC_t/t$. 
Many lower bounds in data streams are proven already with two players.
However, it is known that for some functions more  players are needed to obtain stronger lower bounds, such as for estimating the frequency moments in insertion only streams (see, e.g.,
 \cite{BarYossef:2004eya,wz12} and references therein). 

One cannot help but ask {\it how powerful is communication complexity for proving data stream lower bounds}? Another natural question is: {\it for a given function $f$, which number $t$ of players should one partition the stream into}? 
Yet another question is regarding the input distribution -- should it be a product distribution for which the inputs to the players are chosen independently, 
or should the inputs be drawn from a non-product distribution to obtain the best space lower bounds? 
Since we are interested in the limits of using $t$ players for establishing lower bounds for data stream algorithms, we allow the original $k$ inputs (which correspond to the $k$ elements in a stream) to be partitioned in the worst possible way for a $t$-player communication protocol, as this will give the strongest possible lower bound.

\subsection{Our Results}
In this paper we study these communication questions and their connections to data streams.

We first make the simple observation that for non-product input distributions, the communication complexity can be arbitrarily smaller if we partition the $k$ inputs into $t < k$ players. 
Indeed, consider the $k$-player set disjointness problem in which the $i$-th player, $1 \leq i \leq k$, has a set $S_i \subseteq [n]$, 
where for notational simplicity we define $[n]\eqdef \set{1,2,\dots, n}$ for $n\in \N$. 
The input distribution satisfies the promise that either (1) $S_i \cap S_j = \emptyset$ for every $1\le i < j \le k$, or (2) there is a unique item $a \in [n]$ such that $a \in S_i$ for all $i \in [k]$, and for any other $a' \neq a$, there is at most one $i \in [k]$ for which $a' \in S_i$. It is well-known that the randomized communication complexity of this problem is $\bigomega{n/k}$ \cite{BarYossef:2004eya,g09,jayram2009hellinger}, and that the bound holds even for multiple rounds of communication and when players share a common blackboard. 
However, if we look at $t < k$ players and an arbitrary, even if the worst-case mapping of the input sets $S_1, \ldots, S_k$ to the $t$ players, then by the pigeonhole principle there exists a player who gets two input sets $S_i, S_j$ with $i \ne j$. 
Now this player can locally determine the output of the function by checking if $S_i \cap S_j=\emptyset$. 
Thus with $t < k$ players the problem is solvable using $\bigo{1}$ bits per player. This simple argument shows that for non-product distributions, there can be an arbitrarily large gap between the $k$-player and the $t$-player worst-case-partitioned randomized communication complexities. 
Note that this example applies to a symmetric problem, meaning that the $k$-player set disjointness problem is invariant under any one-to-one assignment of $x_1, \ldots, x_k$ to the $k$ players.

Perhaps surprisingly, and this is one of the main messages of our work: for symmetric functions and product input distributions,

we show that for any $t < k$, for deterministic one-way communication complexity or randomized one-way communication complexity with error probability $1/\poly(k)$, that is, the gap between the $k$-player and $t$-player communication complexities is at most a multiplicative $\bigo{1}$ factor in maximum message length, or the maximum communication from a single player, and $O(k)$ in total communication. Further, this gap is tight, as there are problems for which the input distribution is a product distribution, and the $t$-player communication with $1/\poly(k)$ error probability is $\bigo{\log k}$ for constant $t=\bigo{1}$, while the $k$-player communication with $1/\poly(k)$ error probability is $\bigomega{k \log k}$. 

Thus, the answer for product input distributions is significantly different than what we saw for non-product distributions, even for symmetric functions. 

We also show that for protocols with constant error and under product input distributions, the gap is at most a multiplicative $O(\log k)$ factor in message length and $O(k \log k)$ in total communication. Further, we show that there exists a symmetric function and input distribution which is product on any $k-1$ out of $k$ inputs, for which this gap is best possible. We leave open the question of the existence of a symmetric function and product input distribution (on all $k$ inputs rather than $k-1$ out of $k$) which realizes this gap for constant error protocols.

One takeaway message from our results is that when showing space lower bounds for data stream algorithms computing symmetric functions on product distributions, by looking at $2$-player communication complexity (which is by far the most common communication setup), there is only an $O(1)$ factor loss for error probability $1/\poly(k)$ protocols, and an $\bigo{\log k}$ factor loss for constant error protocols. 

However, for non-product distributions, which are often needed to show hardness of approximation in data streams (such as for the frequency moments \cite{BarYossef:2004eya}), one may need to use as many as $k$ players in order to obtain a non-trivial lower bound from communication complexity. 

\subsubsection{Data Stream Lower Bounds:}
\label{para:Hamming estimation}
As a key application of our lower bound techniques, we provide a space
lower bound for $(1\pm\gap)$-approximating the \emph{Hamming norm} in the strict turnstile model.
This problem, which is also known as the \emph{$L_0$ norm estimation} and denoted by $\hest_\gap$, requires estimating $\norm{\mathbf{x}}_0  \eqdef |\set{i \st x_i \ne 0}|$ of a vector $\mathbf{x} = (x_1, \ldots, x_N)$ and outputting an estimate $\widetilde{F}$ for which $(1-\gap)\norm{\mathbf{x}}_0\le \widetilde{F} \le (1+\gap)\norm{\mathbf{x}}_0 $ with constant probability. 
The vector $\mathbf{x}$ is initialized to all zeros and undergoes a sequence of $m$ updates each of the form 
$(i,v)\in [N]\times [\pm M]$, 
where $[\pm M]\eqdef \set{0,\pm 1,\dots,\pm M}$ and each update $(i,v)$ causes $x_i\gets x_i+v$.
In the strict turnstile model $x_i \geq 0$ holds for all $i$ and at all points in the stream.
% , and at the end of the stream $x_i \le \poly(N)$.
We obtain an $\bigomega{\gap^{-2}\log(N)\log\log(mM)}$ bits of space lower bound for  $(1\pm\gap)$-approximating the {Hamming norm}.
This lower bound matches the best known upper bound $\bigo{\gap^{-2}\log(N)\left(\log(1/\gap)+\log \log (mM)\right)}$ \cite{KNW10} for any $\gap\ge 1/\polylog(mM)$. Note that $\gap \geq 1/\polylog(mM)$ is required in order to obtain polylogarithmic space, and so is the most common setting of parameters.

Perhaps surprisingly, 
there is an upper bound of $\bigo{\gap^{-2}\log(mM)}$ bits of space for $(1\pm\gap)$-approximating
$L_p$ for $p>0$ \cite{knw10b}  (improving an earlier $\bigo{\log^2 N}$ bound of \cite{i06}; see also a time-efficient version in \cite{knpw11}), 
and thus we provide 
a strict separation in the complexities for $p = 0$ and $p > 0$. 

The Hamming norm has many applications, as it corresponds to estimating the number of distinct values, and can be used to estimate set union and intersection sizes (see \cite{CDIM03} where it was introduced).

\paragraph*{Lower Bounds in the Learning Augmented Setting}

Recently, there has been a growing interest in using machine learning to infer information about the stream that would be useful for solving certain problems in the streaming setting. In this learning augmented setting, we have access an oracle (which in practice would have some degree of error and could be implemented with machine learning).  Learned oracles have been used to develop improved algorithms for various problems, including frequency estimation \cite{hsu2018learningbased}, caching \cite{lykouris2018competitive}, scheduling \cite{mitzenmacher_scheduling_2020}, frequency moments \cite{Jiang2020Learning-Augmented}, and more. A fairly comprehensive survey of learning augmented algorithms can be found here: \cite{DBLP:books/cu/20/MitzenmacherV20}

In our setting, the oracle provides an additional operation: we can give the oracle a coordinate, and the oracle will tell us whether the frequency of this coordinate at the end of the stream is at least $T$ for a \emph{threshold} $T$. We refer to this oracle as the \emph{heavy hitters oracle}. Approximate heavy hitter oracles have been used for frequency estimation \cite{hsu2018learningbased}.

 We derive a novel method to prove space lower bounds even with a perfect heavy hitters oracle. We use this method to prove a lower bound of $\bigomega{\gap^{-2}\log(N)\log\log(T)}$ for approximating the $L_0$ norm, which is optimal when $T = 2^{\poly(1/\epsilon)}$ as it matches the upper bound in \cite{Jiang2020Learning-Augmented}.

\subsection{Technical Overview}
We first illustrate the 
idea behind showing there is no gap between $k$-player and $2$-player
deterministic one-round communication complexity. The first player $P_1$
of the $k$-player protocol pretends to be Alice, the first player of the $2$-player protocol, 
to create the message $m_1$ as Alice would do
and sends it to the second player $P_2$ of the $k$-player protocol. 
Having received this message $m_1$,
$P_2$  enumerates over all possible
inputs of $P_1$ until finding one
which would cause $P_1$ to send $m_1$. 
Since the protocol is deterministic and it evaluates a function defined on a product domain, meaning that it is a total function on a domain of the form $S_1 \times S_2 \times \cdots \times S_k$, 
the function value must be the same as long as $P_1$'s input results in the same message $m_1$ to be sent.
So $P_2$ can arbitrarily pick one of those inputs as his guess for $P_1$.
Now $P_2$ has a guess $x$ for $P_1$'s input together
with his own input $y$, and $P_2$ can simulate Alice in the $2$-player protocol. 
This is feasible because the $2$-player protocol works under any partitioning of the inputs. 
Then $P_2$ sends to the third player $P_3$ the message that Alice would send to Bob in the $2$-player protocol, 
given that Alice had input $(x,y)$.
In case when every player $P_i$ cannot figure out how many input items have been processed from his own input and the received message $m_{i-1}$, which is important for his simulation of the $2$-player protocol,
an additional logarithmic-many-bits index carrying this piece of information should be passed together with the simulated messages.
In this way, the entire $k$-player protocol can be simulated
and the per player communication equals to the communication of the $2$-player protocol between Alice and Bob, sometimes plus the additional logarithmic many bits for the index. 
Moreover, both protocols are deterministic.

For the randomized case with a product input distribution, 

we first consider $2$-player protocols with error probability $1/\poly(k)$. 

We would like to run the same simulation as for deterministic protocols, 
except now it is unclear how the second player $P_2$ can reconstruct a valid input $x$ for the first player $P_1$ from the first message $m$. 
A natural thing would be for $P_2$
to choose the input $x=x_m$ to $P_1$ for which the probability
of sending $m$, given that $P_1$'s input is $x_m$, is greatest. 
This is not correct though, since the overall probability of $P_1$ holding $x_m$ and sending $m$ may be less than the $1/\poly(k)$ error bound  
and the protocol could afford to be always wrong on such a combination of $x_m$ and $m$.
Thus we need some balancing between two probabilities:  
i) the first player $P_1$ sends $m$ on input $x$;
and ii) the protocol output is correct 
given that $P_1$ 
has input $x$ and sends $m$. 

The above naturally suggests that we should impose an
input product distribution $\mu$. Then it must be that for a good
fraction of $x$, weighted according to $\mu$,
the $k$-player protocol is correct when the first player has input $x$ and sends message $m$. 
Thus we can sample $x$ from the conditional distribution on $\mu$
given that message $m$ is sent. Here, for correctness, it is crucial
that $\mu$ is a product distribution; this ensures for most settings of
remaining player's inputs (weighted according to $\mu$), for most choices
of $x$ (weighted according to $\mu$) giving rise to $m$, the function evaluated on the inputs is the same, and $x$ can be sampled independently of remaining inputs. 
Once we have sampled $x$, and given that the second player has private input $y$ in the $k$-player protocol, we can then have
the second player pretend to be Alice of a randomized $2$-player
protocol with input $(x,y)$, similar to the deterministic case. 
Ultimately,  
we will show that under distribution $\mu$ we obtain a protocol with
total communication at most $\bigo{k}$ times that of the 
$2$-player protocol with error probability $1/\poly(k)$. The maximum message length, which is an important resource measure in our setting, blows up by at most an $\bigo{1}$ multiplicative factor times that of the $2$-player
protocol, 
where the factor $k$ comes from 
the number of invocations of the $2$-player protocol.

We illustrate the optimality of the randomized reduction above by
looking at the \sumequal problem studied by Viola \cite{Viola:2013we}:
in this problem each of $k$ players holds an input $x_i \mod p$, where
$p = \bigtheta{k^{1/4}}$ is a prime, and they wish to determine whether
$\sum_i x_i = 0$ or $1 \mod p$. Viola shows
this problem has randomized communication complexity $\bigtheta{k \log k}$, 
for both randomized protocols with constant error probability as well
as deterministic protocols (and thus also randomized protocols with
$1/\textrm{poly}(k)$ error probability). Moreover, for randomized protocols
with $1/\textrm{poly}(k)$ error probability, Viola's $\Omega(k \log k)$
lower bound holds even for a product distribution on the inputs (where
if $\sum_i x_i \mod p \notin \{0,1\}$ the output can be arbitrary).  
We observe that under any partition of the inputs into $2$-players
Alice and Bob,
the problem can be solved with $\bigo{\log k}$ bits with probability
$1-1/\textrm{poly}(k)$ just by running an equality test on the sum
modulo $p$ of Alice and the negated sum modulo $p$ of Bob. Thus, this
illustrates that the factor $O(k)$ gap in total communication for protocols for product input 
distributions with $1/\textrm{poly}(k)$ error probability is {\it optimal}. 

On the other hand, for constant error protocols and a product input 
distribution, there is a 
$2$-player $\bigo{1}$ bit upper bound in the public coin model which comes
from running an equality test with constant error probability
(since we measure error with respect to an input distribution, equality
has an $O(1)$ upper bound with constant error). 

We note that the $k$-player protocol has communication $\bigomega{k \log k}$
for constant error protocols, 
which gives the $\bigomega{k \log k}$ factor gap we claimed. The only
downside is that the $\bigomega{k \log k}$ lower bound  holds for an input 
distribution which is product
on $k-1$ out of $k$ players, rather than all $k$ players. We leave it as an open question
to give an optimal separation for product input distributions for 
constant error probability.

Given the importance of Viola's problem in showing separations, 
we next show a {\it direct sum theorem} for his problem,
showing its communication complexity increases to $\bigomega{kr \log k }$
for solving a constant fraction of $r$ independent copies.  

To show the direct sum theorem for Viola's problem, 
one issue is that, unlike for two players where the technique
of {\it information
complexity} often provides direct sum theorems, 
for $k$-players the analogues are much weaker. A natural route would be
to take Viola's {\it corruption bound}, argue it implies a high
information bound, and then apply standard direct sum theorems for
information. This approach does not give an information cost lower
bound on private coin protocols, though one can fix it for two players
using \cite{bg14}, which improves upon a bound in \cite{BBKLSV16}.
However, for $k$ players similarly strong
bounds are unknown. 
Another natural approach is to use the fact
that if a problem has a corruption bound, then one immediately has 
a direct sum for it \cite{BPSW05}. Again though, this is only for two players or
the \emph{number on forehead} model, and not for 
our setting. 

Instead, our proof is inspired by Viola's rectangle argument 
for a single copy of the \sumequal problem, 
where each rectangle, restricted to the first
$k-1$ players, is a product distribution on 
which the protocol generates a message to the $k$-th player. 
We use a rectangle argument on multiple copies where the output 
is now a binary vector instead of a single bit.
The main obstacle is that we must consider the Hamming distance between the protocol output and the correct answer in a vector space,
which is much more involved than studying the error probability for a 
single instance.
The intuition of our proof is that for every large rectangle,
there must be linearly many copies that appear (almost) uniformly 
random in the last player's view.
The above argument is fairly intricate, and involves several levels of 
conversion: 
i) a large rectangle implies large conditional entropy in many players' inputs;
ii) the large entropy of all copies implies we have min-entropy at least $1$ on many copies;
iii) a random variable of min-entropy at least $1$ can always be  decomposed into a convex combination of uniform distributions over two elements;
iv) the summation of sufficiently many independent random variables that are each drawn from a uniform-over-two-element distribution turns out to be nearly uniform,
and hence many \sumequal copies look uniform to the last player.

Thus, the last player can hardly outperform a random guess. 
Note that it is insufficient to prove uniformity for many copies individually 
(which is not too hard using the same idea as in Viola's proof), 
since such a situation could be simulated with a much smaller rectangle 
with very small error.
We instead perform our rectangle argument inductively to show most 
copies appear almost uniform, even if conditioned on previous copies.

This direct sum technique has further applications. One application
is to proving a lower bound for approximating the Hamming norm 
in a strict turnstile stream. 
Using a result of \cite{AHLW16}, to show lower bounds for streaming algorithms in the
strict turnstile model, it suffices to show lower bounds in the simultaneous communication model, where each player simultaneously sends a linear sketch to a referee who outputs the answer. 
To get the desired direct sum property, we have a chain of reductions leading to the \sumequal problem which we compute the information complexity of.

Specifically, we consider a composition of the Gap-Orthogonality problem on top of the \sumequal instances as well as an augmented index version of the composed problem. When we compose these problems, each coordinate of the Gap-Orthogonality problem becomes a \sumequal instance, and we show that in order to solve Gap Orthogonality, we must solve most of the \sumequal instances. Thus. we can use a direct sum to bound the information cost of the composed problem in a similar manner as in \cite{wz12}. We then prove that approximating the Hamming norm reduces to the augmented index version of this, which allows us to bound its communication complexity and accordingly its streaming complexity.

In the augmented problem 
we additionally give a referee an index $i$ and the answers to all copies $j$, with $j > i$. 
Similar augmentation has been studied for $L_p$-norms \cite{knw10b}. 
This allows us to reduce our communication problem to Hamming norm approximation, and ultimately prove our data stream lower bound.

%%%%%%%%%%%%%%%%%%%%%%%%%%%%%%%%%%%%%%%%%%%%%%%%%%%%%%%%%%%%%%%%%%%%%%%%%%%%%%%%
%%%%%%%%%%%%%%%%%%%%%%%%%%%%%%%%%%%%%%%%%%%%%%%%%%%%%%%%%%%%%%%%%%%%%%%%%%%%%%%%
\section{Preliminaries}
\label{sec:pre}

A function $f:\alphabet^k \to \range$ is called a \emph{$k$-party symmetric function} if for every $(x_1,x_2,\dots,x_k)\in \alphabet^k$ and for every permutation $\sigma$ over $\set{1,2,\dots,k}$,
there is 
$f(x_1,\dots,x_k) = f\left(x_{\sigma(1)},  \dots, x_{\sigma(k)}\right).$

	A $k$-dimensional vector space $S$ is called a \emph{product space} if it can be represented as $S=S_1\times S_2 \times\dots\times S_k$.
	A distribution $\mu$ is called a \emph{product distribution} if it is obtained by taking the product of $k$ independent distributions, i.e., $\mu=\mu_1\times \mu_2 \times \dots \times \mu_k$.

In the $t$-player communication complexity model, there are $t$ computationally unbounded players, e.g., $P_1,\dots,P_t$, required to compute a function $f:X_1\times\dots\times X_t\to Y$, where $f$ is usually a $t$-party symmetric function.
Each player $P_i$ is given a private input $x_i\in X_i$ and follows a fixed protocol to exchange messages.
For every input $(x_1,\dots,x_t)$, the message transcript is denoted by $\Pi_t(x_1,\dots,x_t)$ when all players follow the protocol $\Pi_t$
(when $\Pi_t$ is randomized, $\Pi_t(x_1,\dots,x_t)$ is a random variable taking probabilities over players' random coins).
A deterministic protocol $\Pi_t$ computes $f$ if there is a function $\Pi_{out}$ such that $\Pi_{out}\left(\Pi_t^{(t)}(x_1,\dots,x_t), x_t\right) \equiv f$,
where $\Pi_t^{(t)}(x_1,\dots,x_t)$ denotes $P_t$'s view under the execution of $\Pi_t$ on input $(x_1,\dots,x_t)$ and for simplicity we let $\Pi_{out}\left(x_1,\dots,x_t\right)\eqdef \Pi_{out}\left(\Pi_t^{(t)}(x_1,\dots,x_t), x_t\right)$.
A $\error$-error randomized protocol $\Pi_t$ for $f$ requires the existence of $\Pi_{out}$ such that for all inputs $(x_1,\dots,x_t)$, $\Pr\left[ \Pi_{out}\left(x_1,\dots,x_t\right)=f(x_1,\dots,x_t) \right]\ge 1-\error$.
The \emph{communication cost} of $\Pi_t$ is the maximum size of $\Pi_t(x_1,\dots,x_t)$ over all $x_1,\dots,x_t$ and all random coins.
The \emph{$t$-player deterministic communication complexity} (resp. \emph{$t$-player $\error$-error randomized communication complexity}), denoted by $\dcc_t(f)$ (resp. $\rcc_{t,\error}(f)$), is the cost of the best $t$-player deterministic (resp. $\error$-error randomized) protocol $\Pi_t$ for $f$.

Given a $k$-party function $f:X_1\times \dots\times X_k\to Y$ and $t<k$, we define $\dcc_t(f)$ and $\rcc_{t,\error}(f)$ 

under a \emph{worst-case partition} of inputs.
That is, let $f_t(z_1,\dots,z_t)=f(x_1,\dots,x_k)$ be defined for every partition $i_0=0\le i_1 \le \dots \le i_t=k$ and $z_j\eqdef (x_{i_{j-1}+1},\dots,x_{i_j})$,
and the $t$-player communication complexity of $f$ is defined with respect to the worst choice of $f_t$, i.e.,  
$\dcc_t(f)\eqdef\max_{f_t}\dcc_t(f_t)$ and $\rcc_{t,\error}(f)\eqdef \max_{f_t}\rcc_{t,\error}(f_t)$.

Given a $t$-party function $f$ and its input distribution $\mu$, we let $\dcc^{\mu}_{t,\error}(f)$
denote the communication cost of the best $t$-player deterministic protocol $\Pi_t$ computing $f$ such that 
$\Pr_{x\follows \mu}\left[ \Pi_{out}(x)\ne f(x) \right] \le \error$.
Similarly we define  $\rcc^{\mu}_{t,\error}(f)$ for randomized protocols.

In the restricted \emph{one-way communication model} \cite{papadimitriou1984communication,ablayev1996lower,kremer1999randomized}, 
the $i$-th player sends exactly one message to the $(i+1)$-st player for $i\in[t-1]$ following $\Pi_t$,
and then $P_t$ announces the output of $\Pi_t$ as specified by $\Pi_{out}$.
Note that in this setting there are only $k-1$ messages sent by $P_1,\dots, P_{k-1}$, and we do not count the final output announced by $P_t$ in the communication in order to best correspond to streaming algorithms.
This is also known as a \emph{sententious} protocol in previous work, e.g., \cite{Viola:2013we}.
We denote the \emph{$t$-player one-way communication complexities of $f$} by $\dccone_t(f)$ and $\rccone_{t,\error}(f)$, respectively.

In the \emph{common reference string model} (aka \emph{CRS model}), 
there is a sequence of public random coins, which is by default a uniformly random binary string, accessible to all players.
The obvious advantage of communication in the CRS model is that players 
have access to the same random string and thus 
save the cost of synchronizing their private coins.

A streaming algorithm is an algorithm that scans the input $(x_1,\dots,x_m)\in\alphabet^m$ as $m$ stream input items in sequence, updates its internal memory of size $s=\littleo{m \log|\alphabet|}$ (i.e., a streaming automaton with $2^s$ states, where the space cost of updating the internal memory is not accounted for), 
and finally outputs a function $f(x_1,\dots,x_m)$ evaluated on all input items. 
If the best deterministic (resp. $\error$-error randomized) streaming algorithm computes $f$ with $s$ bits of memory and $t$ passes over the data stream, then we say the \emph{deterministic} (resp. \emph{$\error$-error}) \emph{streaming complexity} of $f$ is $st$, denoted by $\dsc(f)=st$ (resp. $\rsc_\error(f)=st$).
In a popular and standard setting, a streaming algorithm scans the input stream in a \emph{single pass} and only processes every input item once. 
The necessary amount of memory required by such single-pass algorithms is called the \emph{single-pass deterministic/$\error$-error streaming complexity} and denoted by $\dscone(f)$ and $\rscone_\error(f)$ respectively.

Note that every streaming algorithm can be naturally interpreted as a communication protocol where each party holds some (possibly an empty set of) input items on the stream and the messages capture the memory updates.
The connection between streaming complexity and communication complexity trivially follows 
in the following lemma.

\begin{lemma}\label{lemma:sc-cc}
For every function $f$ and error tolerance $\error$, 
for every $k\in\N$, it holds that: 
\[\dsc(f)\ge \frac{1}{k} \cdot \dcc_{k}(f),\;\;
\rsc_{\error}(f)\ge \frac{1}{k} \cdot \rcc_{k,\error}(f)\]

Furthermore, similar relations hold for 
 single-pass streaming complexities 

versus
 $k$-player one-way communication complexities: 

\[\dscone(f)\ge \frac{1}{k-1} \cdot \dccone_{k}(f),\;\;
\rscone_{\error}(f)\ge \frac{1}{k-1} \cdot \rccone_{k,\error}(f)\]

\end{lemma}

Additionally, we let $\D_{k,\delta,\mu}(f)$ denote the communication complexity of $f$ with $k$ players and $\delta$ error under input distribution $\mu$ and $\ic_{k,\delta}(f)$ denote the information complexity of $f$ with $k$ players and $\delta$ error. We extend the notion of information complexity from \cite{959901} to this setting by summing the information costs over all of the players and allowing some probability of returning an incorrect answer. The following lemma from relates $\ic_{k,\delta}(f)$ to $\rccsim_{k,\delta}(f)$:

\begin{lemma}
For any function $f$,
    $$\ic_{k,\delta}(f) \le \rccsim_{k,\delta}(f)$$
\end{lemma}
This follows from the fact that the mutual information of the message $M$ that a player sends with their input must be smaller than the number of bits in the message.\\

Additionally, $\ic(f)$ is well-behaved in the sense that it satisfies the direct sum property:

\begin{theorem}
For any function $f$ and any positive integer $m$,
$$\ic_{k,\delta}(f^m) \ge m \cdot \ic_{k,\delta}(f)$$
where a $\delta$ probability of failure for $f^m$ is defined to mean a $\delta$ probability of failure on each instance.
\end{theorem}
This follows from the direct sum theorem on two players and no error by grouping all but player $i$ into the referee for each $i$ and summing over the information complexities of the protocols for each $i$. Then, to deal with the $\delta$ probability of error, we simply force the protocols to be deterministic and consider the function only on the values for which it is correct.

Finally, we introduce the linear sketch model of communication. In this setting, we have $n$ players and the only protocols allowed are of the following form:

There is some matrix $A$ such that if player $i$ receives input $x_i$, they compute $Ax_i$ and send the result to the referee. The referee then computes $\sum_{i=1}^n Ax_i$ and uses the result to compute the answer. We denote the randomized communication complexity of a function $f$ in this model by $\rcclin_{k,\delta}(f)$.\\

We  denote the randomized communication complexity of a function $f$ in this model given that the maximum frequency of any coordinate at the end of the stream is at most $T$ by $\rcclinT_{k,\delta}(f)$. Similarly, in the other models, when we bound the frequency of the coordinates, we will write $\DlinT_{k,\delta,\mu}$, $\ic_{k,\delta}^T$, and $\rscT_{k,\delta}$ for distributional complexity, information complexity, and randomized streaming complexity respectively.

%%%%%%%%%%%%%%%%%%%%%%%%%%%%%%%%%%%%%%%%%%%%%%%%%%%%%%%%%%%%%%%%%%%%%%%%%%%%%%%%
%%%%%%%%%%%%%%%%%%%%%%%%%%%%%%%%%%%%%%%%%%%%%%%%%%%%%%%%%%%%%%%%%%%%%%%%%%%%%%%%

\section{Communication Complexity for Functions on Non-Product Spaces}

\begin{theorem}\label{theorem:1-n-separation}
For every $t\ge 2$, there is a $t$-party symmetric function 
$f: D  \to \zo$ 
defined on $D\subseteq \zon= \left(\zo^{n/t} \right)^t$ 
such that for every error tolerance $\error<1/4$, 
$\dccone_{t-1}(f)\le t-1$ but $\rcc_{t,\error}(f)=  \bigomega{ n/t }$. 
In particular, 
as long as $t=\bigo{1}$ is a constant, we have $\dccone_{t-1}(f) =\bigo{1}$ and $\rsc_\error(f) \ge {\frac{1}{t}\cdot\rcc_{t,\error}(f)} = \bigomega{n}$.
\end{theorem}

\begin{proof}
Consider the $t$-party set disjointness problem $\textsc{Disj}_{n/t,t}$ defined as follows:
there are $t$ players $P_1,\dots, P_t$ such that every player $P_i$ holds a private indicator vector $\mathbf{x}_i\in \zo^{n/t}$ which represents a subset of $[n/t]$, 
i.e.,~$\textsc{Disj}_{n/t,t}(\mathbf{x}_1,\dots, \mathbf{x}_t)=\vee_{j=1}^{n/t} \left(\wedge_{i=1}^t x_{i,j}\right)$, where $x_{i,j}$ denotes the $j$-th coordinate of $\mathbf{x}_i$.
We consider the domain $D$ such that the vectors $\mathbf{x}_1,\dots, \mathbf{x}_t\in\zo^{n/t}$ are either
(1) pairwise disjoint, or (2) sharing a unique element $j\in[n/t]$.
% , i.e.~for every fixed $j\in[n/t]$, $\sum_{i=1}^t x_{i,j} \le 1$.
Let $f$ be the function that computes $\textsc{Disj}_{n/t,t}$ on domain $D$.

On the one hand, it is easy to verify that $\dccone_{t-1}(f)\le t-1$. 
Indeed, at least one of the $t-1$ players obtains two distinct indicator vectors and hence can itself decide the output of $f$. 
The communication is $1$ bit per player to pass the result, and hence the total communication is bounded by $t-1$ since there are $t-1$ players.

On the other hand, the $\Omega(n/t)$ lower bound for $\rcc_{t,\error}(f)$ follows from the known lower bound for multi-player set disjointness 
(see \cite{BarYossef:2004eya}, which was improved to optimal in \cite{g09,jayram2009hellinger}). 
The lower bound for $\rsc_\error(f)$ immediately follows by \cref{lemma:sc-cc}.
\end{proof}

%%%%%%%%%%%%%%%%%%%%%%%%%%%%%%%%%%%%%%%%%%%%%%%%%%%%%%%%%%%%%%%%%%%%%%%%%%%%%%%%
%%%%%%%%%%%%%%%%%%%%%%%%%%%%%%%%%%%%%%%%%%%%%%%%%%%%%%%%%%%%%%%%%%%%%%%%%%%%%%%%

\section{Deterministic Communication and Streaming Complexity}

We first show that 2-player one-way communication complexity is equivalent to the streaming complexity of single-pass streaming algorithms in the deterministic setting. In the following theorem, we assume for convenience that $m$ is known to both players.

\begin{theorem}\label{thm:dsc-dcc}
For every symmetric function $f:\alphabet^m\to \range$, $\dccone_2(f)\le \dscone(f)\le \dccone_2(f) + \log m$.
\end{theorem}
\begin{proof}
Obviously, $\dscone(f) \ge \dccone_2(f)$ since a 2-player communication protocol  simulates a streaming algorithm.
It remains to prove $\dscone(f) \le \dccone_2(f) + \log m$.

Suppose the input stream is $(x_1,\dots,x_m)\in\alphabet^m$,
and for every partition into $(x_1,\dots, x_i)$ and $(x_{i+1},\dots,x_m)$ there is a deterministic 2-player one-way  protocol $\Pi_2^i$ computing $f$.
We design the deterministic single-pass streaming algorithm $A$ for $f$ by simulating $2$-player one-way communication protocols under different partitions. 
The memory usage of $A$ is therefore bounded by the maximum communication cost of the simulated $2$-player protocols plus an index in $[m]$ recording the number of processed items.

Notice that when processing the item $x_{i+1}$, $A$ has already processed $x_1,\dots, x_i$ and has $(m_i,i)$ in memory. 
$A$ can thus reconstruct a compatible guess of $x''_1,\dots, x''_i$ that would induce exactly the message $m_i$ as in $\Pi_2^i$, and then sets the memory to be $(m_{i+1},i+1)$ where $m_{i+1}$ is the message sent in $\Pi_2^{i+1}$ when $P_1$ has $(x''_1,\dots, x''_i, x_{i+1})$ and $P_2$ has $(x_{i+2},\dots, x_m)$.
$A$ repeats this process for every $i=1,\dots,m-1$ and at the end it outputs $f(x_1,\dots,x_m)$.

Therefore, we complete the proof with $\dccone_2(f)\le \dscone(f) \le \dccone_2(f)+\log m$.
\end{proof}

Note that the additional index $i$ in the above simulation, which results in the additive $\log m$ term in the upper bound, indicates which $2$-player protocol should be simulated in the reconstruction,
and it is implicitly shared in the $2$-player communication case when $m$ is common knowledge.

When $m$ is not known, the memory used for the index follows any previously agreed upon encoding, which uses $O(\log m)$ space.
For functions that are well-defined for an arbitrary number of input items, e.g. the parity function, this index can be saved, and hence $\dscone(f) = \dccone_2(f)$.

For communication complexity among more players, we establish the following corollary.
\begin{corollary}\label{corollary:dcck-dcc2} 
	For every $k$-party symmetric function $f$,		
	\[ (k-1) \cdot \dccone_2(f) \le \dccone_k(f)  \le  (k-1) \cdot \left(\dccone_2(f) +\log k\right) \]
\end{corollary}
\begin{proof}
	Combining \cref{lemma:sc-cc,thm:dsc-dcc}, it follows that
	\[\dccone_k(f) \le (k-1) \cdot \dscone(f) \le (k-1) \cdot \left(\dccone_2(f) +\log k\right)\]
	The other direction $\dccone_k(f) \ge (k-1) \cdot \dccone_2(f)$ holds by giving $z_j=\emptyset$ to every player $j\in\set{2,\dots, k-1}$ in the $k$-player case, when the problem degenerates to $2$-player communication but the same message has to be passed $k-1$ times. 
\end{proof}

Such a linear separation naturally extends to the communication complexity of $t$-player versus $k$-player protocols, as long as $2\le t<k$.
Thus, the deterministic communication complexity grows \emph{linearly} in the number of parties.

We remark that if every player must get a non-trivial input, i.e., at least one input element to the function, the linear growth remains for some but not all problems.
For example, the communication complexity of the parity of $k$ bits is linear in the number of players.
However, to decide whether $k$ elements in $[k]$ are distinct, the $2$-player protocol requires communication $\log\binom{k}{k/2}\approx k-\log\sqrt{k}$, whereas the $k$-player worst-case communication grows sublinearly, i.e. for $k$ players the communication is no more than $\sum_{i=1}^{k-1} \log\binom{k}{i} \ll (k-1)\cdot \log \binom{k}{k/2}$.

\section{Communication Complexity for Functions on a Product Space}
\label{sec:total}
\subsection{Separations for Randomized Communication Complexity}

In this section, we consider the communication cost of randomized multi-player protocols defined on product input distributions and present a $k\log k$ versus $t\log t$ separation between $k$-player and $t$-player communication complexity.

First we introduce the \sumequal problem (as used in Viola's work \cite{Viola:2013we}).

\begin{definition}
The \emph{$k$-player \sumequal over integers}, denoted by $\sumequal_k$, 
requires deciding whether $\sum_{i =1 }^{k} x_i = 0$,
where each player $P_i$ is given an integer $x_i$ as his private input together with the integer $k$ as public input shared by all players.
In the CRS model, an additional public random string is also known to all players.
The \emph{$k$-player \sumequal over $\Z_m$}, denoted by $\sumequal_{k,m}$, 
is defined similarly 
as $\sumequal_{k}$,
except that the input items are drawn from $\Z_m$ and the summation is over $\Z_m$,
for a publicly known $m$.
\end{definition}

\begin{lemma}[\cite{Viola:2013we}, Theorem~15 and Theorem~29]
\label{lemma:Viola}
For every $k\in\N$, $0\le \error\le 1/3$, 
and in the CRS model, 
the $k$-player $\error$-error communication complexity of $\sumequal$ satisfies:

\begin{alphaenumerate}
\item For every $m\in \N$, $\rccone_{k,\error}(\sumequal_{k,m}) = \bigo{k\log (k/\error)}$.

\item For every prime $p\in(k^{1/4},2k^{1/4})$, $\rcc_{k,\error}(\sumequal_{k,p}) = \bigomega{k\log k}$.%
\footnote{Viola's states the lower bound for constant $\error$, but it naturally holds for smaller $\error$ (sometimes not tight).}
\end{alphaenumerate}

In particular, $\rcc_{k,\error}(\sumequal_{k,p})=\bigtheta{k\log k}$ in the CRS model if $\error=\bigomega{1/\poly(k)}$.
\end{lemma}

We remark that Viola's lower bound for $\sumequal_{k,p}$ is proved for a non-product distribution $\mu_H$ whose support covers exactly a $2/p$ fraction of the whole (product) input  space.
Thus if a $k$-player protocol solves $\sumequal_{k,p}$ with error $\error\le 1/k$ on a uniform distribution $\mu$ over the whole input space,
then its error with respect to $\mu_H$ is bounded by $\frac{1/k}{2/p}< k^{-3/4}$.
Notice that the two player version of $\sumequal_{k,p}$ degenerates to testing equality over $\Z_p$ whose upper bound is $\bigo{\log(1/\error)+\log\log k}$,
see more details in \cref{sec:testing equality}.
By \cref{lemma:Viola}, the $\bigomega{k}$ separation in \cref{corollary:lower bound} naturally follows.

\begin{corollary}\label{corollary:lower bound}
	For every prime $p\in(k^{1/4},2k^{1/4})$ and $\error\le 1/\poly(k)$, there is a product distribution $\mu$ such that $\rcc_{k,\error}^{\mu}(\sumequal_{k,p}) = \bigomega{k\log k}$,
	$\rccone_{2,\error}(\sumequal_{k,p}) = \bigo{\log k}$.
\end{corollary}

For a larger error tolerance, say $\error$ is a constant, we have a stronger separation between $k$-party communication and $t$-party communication. However, the hard distribution is slightly non-product, that is, it is a product distribution on any $k-1$ out of the $k$ players. 

\begin{corollary}\label{theorem:log-separation}
For every $k\in \N$, there is a $k$-party symmetric function $f$ such that

\begin{alphaenumerate}
	\item For any product distribution $\mu$, 
              for every $2\le t\le k$ and $0\le \error \le 1/3$, 
	$\rccone_{t,\error}^{\mu}(f)=\bigo{t\log(t/\error)}$.
	In particular, $\rccone_{2,\error}^{\mu}(f)=\bigo{\log(1/\error)}$.

	\item There exists a distribution $\mu_H$, which is product on any
              $k-1$ out of $k$ players, for which 
              $\rcc_{k,\error}^{\mu}(f)= \bigomega{k\log k}$ 
              as long as $\error \le 1/3$.	
\end{alphaenumerate}
\end{corollary}

For $\error\ge 1/\poly(t)$,  
the gap between $\rcc_{k,\error}^{\mu}(f)$ and $\rccone_{t,\error}^{\mu}(f)$ is bounded as below:
\[\rcc_{k,\error}^{\mu}(f) \;\big/\; \rccone_{t,\error}^{\mu}(f) = \bigomega{\frac{k\log k}{t\log t}}\]

The outline of the proof of \cref{theorem:log-separation} was given in Section \ref{sec:intro}.
That is, the upper bound in part (a) follows from applying $k = t$ in 
the first part of Lemma \ref{lemma:Viola},
while the lower bound in part (b) follows from the second part of Lemma \ref{lemma:Viola}.

\subsection{Tightness of the Communication Complexity Separation}

The following theorem and corollary show tightness of our separations.

\begin{theorem}\label{theorem:logk-tight}
For every $k$-party function $f: \alphabet^k\to \range$, product distribution $\mu$ over $\alphabet^k$, and error tolerance $\error <1/3$, 

then the following holds:
\begin{align*}
\rccone_{k,\error}^{\mu}(f) =& \begin{cases}
	\bigo{k\log k}	 \cdot  \rccone_{2,\error}(f) & \mbox{if $\error={\bigomega{1}}$}\\
	\bigo{k}	 \cdot  \rccone_{2,\error}(f)+\bigo{k\log k}	& \mbox{if $\error \le 1/k^{\bigomega{1}} $}
\end{cases}  
\end{align*}
When $\delta > 0$, we have that $\rccone_{2,\error}(f) = \Omega(\log k)$ and thus the following holds:
\begin{align*}
{\rccone_{k,\error}^{\mu}(f)}\; \big/\; {\rccone_{2,\error}(f)} \le & 
\bigo{k \cdot \left(1+\frac{\log k}{\log(1/\error)}\right) }
= \begin{cases}
	\bigo{k\log k}	  & \mbox{if $\error={\bigomega{1}}$}\\
	\bigo{k}		& \mbox{if $\error=1/k^{\bigomega{1}}$}
	\end{cases}
\end{align*}

\end{theorem}

\begin{proof}

First we let $\Pi_0$ be the optimal $\error$-error $2$-player one-way protocol $\Pi_0$ that computes $f$ with communication $C=\rccone_{2,\error}(f)$,
and construct a new protocol $\Pi_2$ by taking 
the majority of $M$ independent parallel copies of $\Pi_0$ such that $\Pi_2$ has error $\varepsilon = \error^2/(16k^2)$ and communication ${C M}$.
Recall that $\Pi_0$ has $\error<1/3$,  it suffices to let $t$ and $M$ be defined as in Lemma~\ref{lemma:amp} below:
\begin{align}\label{equ:M}
t &= \ceil{\log\left(\error/(16k^2) \right)/\log\left(4\error(1-\error) \right)}\\
M &=1+2t=1+2\ceil{\frac{\log(1/\error)+2\log k+4}{\log(1/\error)+\log(1/(1-\error))-2}}=\bigtheta{1+\frac{\log k}{\log(1/\error)}}	
\end{align}

\begin{lemma}\label{lemma:amp}
Let $t\in \N$ and $X_1,X_2,\dots, X_{2t+1}$ be i.i.d. binary random variable such that $\Pr[X_i=1]=\error<1/2$ for every $i\in[t]$,
and let $Y=\mathbf{Majority}\set{X_1,\dots, X_{2t+1}}$ be the majority of all $X_i$'s.
Then $\Pr[Y=1]\le \gap$ as long as $t\ge \log(\gap/\error)/\log(4\error(1-\error))$.	
\end{lemma}

\begin{proof}
For $0<\error<1/2$ and $t\ge \log(\gap/\error)/\log(4\error(1-\error))$,
we have 
\begin{align*}
	\Pr\left[ Y=1 \right] &= \Pr\left[ |\set{i\st X_i=1}|\ge t+1 \right]
	\\
	&= \sum_{j=t+1}^{2t+1} \binom{2t+1}{j} \error^j(1-\error)^{2t+1-j}\\
	&\le \sum_{j=t+1}^{2t+1} \binom{2t+1}{j} \error^{t+1}(1-\error)^{t}\\
	&=\frac{2^{2t+1}}{2} \cdot \error^{t+1}(1-\error)^{t}
	=\left( 4\error(1-\error) \right)^t \cdot\error\\
	&\le \frac{\gap}{\error} \cdot\error =\gap
\end{align*}

The first inequality holds because $\error<1/2$ and hence $\error^j(1-\error)^{2t+1-j}\le \error^{t+1}(1-\error)^{t}$ for $j\ge t+1$.	
The second inequality holds because $4\error(1-\error)<1$ for $\error<1/2$, and $\left( 4\error(1-\error) \right)^t\le \left( 4\error(1-\error) \right)^{\log(\gap/\error)/\log(4\error(1-\error))} = \gap/\error$. 
Thus, we have proved that $\Pr[Y=1]\le \gap$ for $t \ge {\log(\gap/\error)/\log(4\error(1-\error))}$.
\end{proof}

Note that $\Pi_2$ is still a $2$-player one-way protocol but has communication ${CM}$. 
Furthermore, we remark that $CM=\bigomega{\log k}$ for $\error>0$, since the error probability must be $\error\ge 1/2^C$ if it is not zero, and hence $M=\bigtheta{1+\frac{\log k}{\log(1/\error)}}=\bigomega{1+\frac{\log k}{C}}$.

Second we prove that for every product input distribution $\mu$ over $\alphabet^k$, the $k$-party function $f$ can be evaluated by a randomized $k$-player one-way protocol $\Pi_k$ with communication $\bigo{k\cdot (CM+\log k)}$ and error $\error/2$ with respect to $\mu$.
The idea is that given the product input distribution $\mu$, each player $P_i$ acts as follows:
\begin{enumerate}
	\item $P_i$ first assumes that the received message $m_{i-1}$ from $P_{i-1}$ will lead to a correct answer with probability $\ge 1-\frac{\error}{4k}$ with respect to $\mu$.

	\item $P_i$ samples a possible input $x'_1,\dots,x'_{i-1}$ of previous players $P_1,\dots, P_{i-1}$, such that if Alice gets input $(x'_1\dots, x'_{i-1})$ and sends $m_{i-1}$, then with probability $\ge 1-\frac{\error}{4k}$ the protocol $\Pi_2$ leads the correct answer. 
	The probability is taken over internal randomness and Bob's input following the marginal distribution of $\mu$ on the remaining players (here we use the condition that $\mu$ is a product distribution).

	\item Finally, $P_i$ sends a message $(m_i,i)$ of length $CM+\log k = \bigo{CM}$, where $m_i$ is the message that Alice would send in $\Pi_2$ when her input is $(x'_1,\dots,x'_{i-1},x_i)$.
\end{enumerate}

By a union bound the error probability of $\Pi_k$ is bounded by $k\cdot (\frac{\error}{4k}+\frac{\error}{4k})<\error/2$ with respect to $\mu$.
The fact that $\mu$ is a product distribution is used in the second step where the sampling process relies on that previous players' inputs are independently distributed from that of future players.

Thus we finish the proof and conclude that 
$\rccone_{k,\error}^{\mu}(f)\le \bigo{kCM}$.
\end{proof}

Notice that in the proof of Theorem~\ref{theorem:logk-tight}, every message in $\Pi_k$ has the length bounded by $\bigo{CM}$, which gives an upper bound for the single-pass streaming complexity.

\begin{corollary}\label{corollary:separation upper bound}
For every $k$-party function $f$ and product input distribution $\mu$, and for every $\error < 1/3$, 
$\rsc_{\error}^{\mu}(f) \le\rscone_{\error}^{\mu}(f) \le  \bigo{1+\frac{\log k}{\log(1/\error)}} \cdot \rccone_{2,\error}(f).$
\end{corollary}

\section{A Direct Sum for Viola's Problem}
We next turn to our direct sum theorem for Viola's problem, which is a crucial
building block for our streaming application. 
Note that the theorem is proved for $\error < 1/9$, but lower bounds for large error tolerance such as $\error=1/3$ can be obtained using a standard error amplification argument.

	\begin{theorem}\label{thm:A1}
		Let $F:\left(\Z_p^m\right)^k \to \zo^m$ be the $k$-party function computing $m$ independent copies of $\sumequal_{k,p}$, where $p$ is a prime between $k^{1/4}$ and $2k^{1/4}$. 
		For every error tolerance $\error\in(0,1/9)$,  
        we say a protocol $\Pi$ is \emph{correct with probability $1-\error$}
		if there is a reconstruction
		function $G$ such that for every fixed $i \in [m]$ and input $x\in \left(\Z_p^m\right)^k$, $G\big(i, \Pi_{out}(x)\big)$ equals the
		output of the $i$-th instance of $\sumequal_{k,p}$ with probability at least
		$1-\error$, over the internal randomness of $\Pi$. 
		Then the communication cost of any $\Pi$ which is correct with probability $1-\error$, is $\bigomega{m k\log k}$.
	\end{theorem}

% 
%%%%%%%%%%%%%%%%%%%%%%%%%%%%%%%%%%%%%%%%%%%%%%%%%%%%%%%%%%%

	\begin{proof}%[Proof sketch of Theorem~\ref{thm:A1}]
		For simplicity of notation in the proof, we flip the output of $F$, so that it outputs $0$ if the input to the corresponding $\sumequal_{k,p}$ instance sums to $0$ in $\Z_p$, and $F$ outputs $1$ on instances with summation other than $0$.

		Let $\Pi$ be an $\error$-error randomized protocol for $F$, and let $\Pi_{out}\left(x\right)$ denote the output of $\Pi$ on input $x$.
		Here by ``the $\error$-error protocol'' we mean that the expected error rate of $\Pi$ is bounded by $\error$, since both $\Pi_{out}(x)$ and $F(x)$ are binary vectors in $\zo^m$.
		Therefore, 
		\begin{align*}
			\Pr_{i\in_R [m]}\left[\Pi_{out}\left(x\right)_i \ne F_i(x) \right]\le \error
		\end{align*}

		where the input to $F$ is partitioned as $x=\left(x^{(1)},x^{(2)},\dots,x^{(m)}\right)\in \Z_p^{m\times k}$ such that $F_i(x)\eqdef \overline{\sumequal_{k,p}}(x^{(i)})$ computes the $i$-th instance of $\sumequal_{k,p}$ for each $i\in [m]$.

		We abuse notation a little in this proof and let $|\cdot|$ denote the \emph{Hamming weight} of a not necessarily binary vector, 
		which measures the number of non-zero coordinates of the vector. Then, 
		\begin{align*}
		\E\left[ \left|\Pi_{out}\left(x\right) - F(x)\right| \right]\le \error m 
		\end{align*}

		To prove that $\rcc_{k,\error}(F) = \max_{x}|\Pi(x)|=\bigomega{m k\log k}$ for the optimal $\error$-error protocol $\Pi$,
		we will deduce a contradiction if $\Pi$ uses $c<\gamma m k\log k$ bits of communication,
		for a constant $\gamma=(1-9\error)/135>0$ and sufficiently large $k$.
		Thus, we can conclude a communication lower bound of $c\ge \gamma mk\log k= \bigomega{m k \log k}$.

		For the purposes of a contradiction, we first convert the randomized protocol $\Pi$ into a deterministic protocol $\Pi'$ that has small error with respect to a specific distribution $\mathcal{H}$. 
		The deterministic protocol $\Pi'$ is obtained by fixing all internal random coins of $\Pi$ so that $\Pi'$ has error rate at most $\error$ for inputs drawn from $\mathcal{H}$.
		\begin{align*}
			\E_{X\follows \mathcal{H}}\left[ \left| \Pi'_{out}(X) - F(X)\right| \right]\le \error m
		\end{align*}

		Since $\Pi'$ can never generate a transcript larger than the communication that $\Pi$ uses in the worst case, i.e., $|\Pi'(X)| \le  \max_x |\Pi(x)| =c$,
		it suffices to prove a communication lower bound for $\Pi'$ on inputs drawn from $\mathcal{H}$.

		By Markov's inequality, we have that for every positive constant $\gap>0$,
		\begin{align}\label{ineq:Pi'}
		\Pr_{X\follows \mathcal{H}}\left[ \left| \Pi'_{out}(X) - F(X)\right| > \gap m \right] \le \frac{\E_{X\follows \mathcal{H}}\left[ \left| \Pi'_{out}(X) - F(X)\right| \right]}{\gap m}\le \frac{\error}{\gap}
		\end{align}

		Now we specify the distribution $\mathcal{H}$.
		Let $G$, $B$ be defined as
		\[
			\left\{
			\begin{array}{c c r}
				G \eqdef &	\big( G_1,\dots, G_{k-1} , &-\sum_{j=1}^{k-1} G_j \big)		\\
				B \eqdef &  \big( B_1,\dots, B_{k-1} , &1-\sum_{j=1}^{k-1} B_j \big)
			\end{array}		
			\right.
		\]
		for uniform and independent $G_i,B_i\in \Z_p$ for every $j\in[k-1]$. 
		Note that: 
		a) $\sumequal_{k,p}(G)=1$, $\sumequal_{k,p}(B)=0$ and hence $F_i(G)=0, F_i(B)=1$;
		b) the first $k-1$ elements of $G$ and $B$, denoted by $G_{-k}$ and $B_{-k}$, follow the same distribution, i.e., the uniform distribution over $\Z_p^{k-1}$.
		For convenience we can write $B=(G_{-k},1+G_k)$.

		Let $H\eqdef G/2+B/2$ be a mixture of $G$ and $B$ and let $\mathcal{H}$ be $m$ independent copies of $H$ as below:
		\[\mathcal{H} \eqdef H^m = \left(G/2+B/2\right)^m\]
		Since $B=(G_{-k},1+G_k)$ and $H=G/2+B/2$, we note that
		\[\mathcal{H} = \sum_{v\in\zo^m} \frac{1}{2^m} (G^m_{-k}, v+G^m_k)=(G^m_{-k}, V+G^m_k),\]
		where $G^m_{-k}$ is uniformly distributed over $\Z_p^{m\times (k-1)}$, $G_k^m$ is a vector in $\Z_p^m$ such that $G_k^m=-\sum_{j=1}^{k-1} G^m_j$, and $V$ is a random variable that is uniform over $\zo^m$, that we will think of as an element in  $\Z_p^m$.
		With the above notation of $\mathcal{H}, V$, we have
		\[F(\mathcal{H})=F(G^m_{-k}, V+G^m_k)=V\]

		To prove the communication lower bound of a deterministic protocol $\Pi'$ that has error probability $\le \error$ w.r.t. $\mathcal{H}$, 
		we recall the following protocol decomposition by monochromatic rectangles, c.f. Claim 24 in  \cite{Viola:2013we}  or Lemma 1.16 in \cite{kushilevitz1997communication}.

		\begin{claim}[\cite{Viola:2013we}, Claim 24]\label{claim: rectangle}
		A $k$-player (number-in-hand) deterministic protocol using communication $\le c$ partitions the inputs into $C\le 2^c $ sets of inputs $R^1,R^2,\ldots, R^C$ such that
			\begin{itemize}
				\item the protocol outputs the same value on inputs in the same set, and

				\item the sets are rectangles: each $R^i$ can be written as $R^i=R_1^i\times R_2^t\times\ldots\times R_k^i$ where $R_j^i$ is a subset of the inputs of Player $j$.
			\end{itemize}
		\end{claim} 

		For every $i\in [C]$ and rectangle $R^i$, we use the notation $R^i_{-j}\eqdef R^i_{1} \times R^i_{2} \times \dots \times R^i_{j-1}\times R^i_{j+1}\times \dots \times R^i_k$ to denote the projection of $R^i$ on to the $k-1$ coordinates except the $j$-th one, for every $j\in[k]$. 
		In particular, $R^i_{-k}\eqdef R^i_{1} \times R^i_{2} \times \dots \times R^i_{k-1}$ denotes the first $k-1$ coordinates.
		Sometimes the index $i$ of rectangle $R^i$ is clear from context, for which we simply write $R$ instead of $R^i$.

		In what follows we show a contradiction when $\Pi'$ has communication $c<\gamma mk\log k$ and hence there are $C\le 2^c<k^{\gamma m k}$ rectangles.
		The argument depends on the following lemma, which essentially guarantees that for every large rectangle, $\Pi'$ is likely to make mistakes on more than $\gap m$ coordinates.

		\begin{lemma}\label{lem: partial uniform}
			For every rectangle $R$ satisfying $\Pr[\mathcal{H}_{-k}\in R_{-k}]\ge \frac{1}{\alpha C}>\frac{1}{\alpha k^{\gamma mk}}$ for which $\alpha= p^{\bigo{1}}$,
			there must be a set $L\subseteq [m]$ such that $|L|=(1-135\gamma) m$ and
			$G_k^{(L)}\cond G_{-k}^m \in R_{-k}$ is  $\frac{|L|}{p}$-close to uniform over $\Z_p^{|L|}$.
		\end{lemma}

		\cref{lem: partial uniform} implies the following claim:

		\begin{claim}\label{claim:ball}
			For every rectangle $R$ on which $\Pi'$ outputs $w\in \zo^m$,
			if $\Pr\left[\mathcal{H}_{-k}\in R_{-k}\right]\ge \frac{1}{\alpha C}$, 
			then for every $u\in R_k$ and for $\gamma,\gap$ satisfying $1-135\gamma \ge 3\gap$, 
			\begin{equation}\label{ineq:ball}
				  \Pr\Big[ \mathcal{H}\in R,\; \left|F(\mathcal{H})-w\right|\le \gap m \Big] 
				< \frac{1}{2} \Pr\left[ \mathcal{H}\in R\right]
			\end{equation} 	
		\end{claim} 

		For compactness of the proof of Theorem~\ref{thm:A1} we defer the proofs of Claim~\ref{claim:ball} and Lemma~\ref{lem: partial uniform}  to the end of this section.

		Let $\widetilde{\mathcal{R}}$ be the set of the $C$ rectangles and $\mathcal{R}\subseteq\widetilde{\mathcal{R}}$ be the set of all large rectangles satisfying $\Pr[\mathcal{H}_{-k}\in R_{-k}]\ge \frac{1}{\alpha C}>\frac{1}{\alpha k^{\gamma mk}}$.
		Then for every rectangle $R\in \widetilde{\mathcal{R}}\backslash\mathcal{R}$, 
		\[\Pr[\mathcal{H}\in R]\le  \Pr[\mathcal{H}_{-k}\in R_{-k}]< \frac{1}{\alpha C} \le \frac{1}{\alpha \big|\widetilde{\mathcal{R}}\backslash\mathcal{R}
		\big|} \]

		Using Claim~\ref{claim:ball}, we have
		\begin{align*}
			&\Pr_{X\follows \mathcal{H}}\left[ \left| \Pi'_{out}(X) - F(X)\right| \le \gap m \right]	\\
			= & \sum_{R\in \widetilde{\mathcal{R}}} %\Pr[\mathcal{H}\in R] 
			\Pr\left[  \mathcal{H}\in R,\; \left|F(\mathcal{H})-\Pi'_{out}(R)\right|\le \gap m \right] \\
			\le & \sum_{R\in {\mathcal{R}}} %\Pr[\mathcal{H}\in R] 
			\Pr\left[  \mathcal{H}\in R,\; \left|F(\mathcal{H})-\Pi'_{out}(R)\right|\le \gap m \right] 
			+ \sum_{R\in \widetilde{\mathcal{R}}\backslash\mathcal{R}} \Pr[\mathcal{H}\in R] \\
			\le & \sum_{R\in {\mathcal{R}}}  \frac{1}{2} \Pr\left[ \mathcal{H}\in R\right]
			+ \sum_{R\in \widetilde{\mathcal{R}}\backslash\mathcal{R}} \Pr[\mathcal{H}\in R] \\
			\le & \frac{1}{2} \sum_{R\in \widetilde{\mathcal{R}}} \Pr[\mathcal{H}\in R] +\frac{1}{2} \cdot \left| \widetilde{\mathcal{R}}\backslash\mathcal{R}\right|\cdot \frac{1}{\alpha \big|\widetilde{\mathcal{R}}\backslash\mathcal{R}
		\big|}
			\le \frac{1}{2}+\frac{1}{2\alpha}
		\end{align*}	

	Combining it with (\ref{ineq:Pi'}), we have 
	\[1-\frac{\error}{\gap} \le \Pr_{X\follows \mathcal{H}}\left[ \left| \Pi'_{out}(X) - F(X)\right| \le \gap m \right] \le \frac{1}{2}+\frac{1}{2\alpha}  
	\implies 1-\frac{2\error}{\gap}\le \frac{1}{\alpha}\]
	
	However, the above inequality cannot be true if we set $\gap =3\error$ and pick a constant $\alpha>3$.
	Let $\gamma\eqdef (1-9\error)/135$ be the constant for which we want to show $c\ge \gamma mk\log k= \bigomega{m k \log k}$.
	Then $1-135\gamma = 9\error \ge 3\gap$ satisfies the condition in \cref{claim:ball} and $\alpha=\bigo{1}$ satisfies the requirement in \cref{lem: partial uniform}.

	Thus we finish the contradiction argument and complete the proof of \cref{thm:A1} with $\rcc_{k,\error}(F)\ge  \gamma mk\log k =\bigomega{m k\log k}$. 
	\end{proof}

	\begin{proof}[Proof of Claim~\ref{claim:ball}]
		Recall that $G'\eqdef  G_k^{(L)}\cond G_{-k}^m \in R_{-k}$,
		$G'$ is $|L|/p$ close to the uniform distribution by Lemma~\ref{lem: partial uniform}. 
		Therefore for every fixed $u\in \Z_p^{|L|}$, 
		\begin{align*}
			& \sum_{v\in\zo^{|L|}: |v|\le \gap m} \Pr\left[ G'=u-v \right]\\  
			=& \frac{1}{2}\left(\sum_{v:|v|\le \gap m}\Pr\left[G'=u-v\right]+\sum_{v:|v|\ge |L|-\gap m}\Pr\left[G'=u-\overline{v} \right]\right)\\
	 		\le & \frac{1}{2}\left(\sum_{v:|v|\le \gap m}\Pr\left[G'=u-v\right]+\sum_{v:|v|\ge |L|-\gap m}\Pr\left[G'={u-v} \right]+\frac{2|L|}{p}\right)\\
	 		< & \frac{1}{2}\sum_{v\in \zo^{|L|}}\Pr\left[G'={u-v} \right]
		\end{align*}
		where the first inequality follows Lemma~\ref{lem: partial uniform},
		and the last inequality holds since as long as $G'$ is close to the uniform distribution and $|L|=(1-135\gamma)m\ge 3\gap m$, there is 
		\[\sum_{v:\gap m<|v|<|L|-\gap m}\Pr\left[G'={u-v} \right] = \bigomega{1}> \frac{|L|}{p}\]

        Recall that $u_L$ and $v_L$ denote $u$ and $v$ restricted to coordinates in the set $L$, $u_{-L}$ and $v_{-L}$ denote $u$ and $v$ restricted to coordinates not in $L$, and $G_k^{(-L)}$ denotes $G_k$ restricted to coordinates not in $L$. We then apply the above inequality and get the following bound relating probabilities on a single coordinate conditional on the rest of the coordinates being contained in the rectangle $R$:
		\begin{align}
			& \sum_{v\in\zo^m: \left|v-w\right| \le  \gap m }\Pr\left[ G^m_k = u-v \Cond G^m_{-k}\in R_{-k}  \right] \notag\\
			\le & \sum_{v\in\zo^m: \left|v_L-w_L\right| \le  \gap m }\Pr\left[ G^m_k = u-v \Cond G^m_{-k}\in R_{-k}  \right] \notag\\
			= & \sum_{v_L\in\zo^{|L|}: \left|v_L-w_L\right| \le  \gap m }\Pr\left[ G^{(L)}_k = u_L-v_L \Cond G^m_{-k}\in R_{-k}  \right] \notag\\
			& \qquad {\Large\boldsymbol{\cdot}} \sum_{v_{-L}\in \zo^{m-|L|}} \Pr\left[ G^{(-L)}_k = u_{-L}-v_{-L} \Cond G^m_{-k}\in R_{-k},G^{(L)}_k = u_L-v_L  \right]  \notag\\
			< & \frac{1}{2} \sum_{v_L\in\zo^{|L|} }\Pr\left[ G^{(L)}_k = u_L-v_L \Cond G^m_{-k}\in R_{-k}  \right] \notag\\
			& \qquad {\Large\boldsymbol{\cdot}} \sum_{v_{-L}\in \zo^{m-|L|}} \Pr\left[ G^{(-L)}_k = u_{-L}-v_{-L} \Cond G^m_{-k}\in R_{-k},G^{(L)}_k = u_L-v_L  \right]  \notag\\
			= &\frac{1}{2} \sum_{v\in\zo^{m} }\Pr\left[ G^{m}_k = u-v \Cond G^m_{-k}\in R_{-k}  \right] \label{ineq:intermediate}
		\end{align}
		\\
		The above inequality (\ref{ineq:intermediate}) implies (\ref{ineq:ball}) since:
		\begin{align*}
			 & \Pr\left[ \mathcal{H}\in R, \left|F(\mathcal{H})-w\right|\le \gap m \right] \\
			=& \Pr\left[ \mathcal{H}_{-k}\in R_{-k} \right]
			\cdot \Pr\left[ \mathcal{H}_k\in R_k, \left|F(\mathcal{H})-w\right|\le \gap m \Cond \mathcal{H}_{-k}\in R_{-k}  \right]\\
			=& \Pr\left[ G^m_{-k}\in R_{-k} \right]
			\cdot \sum_{v\in\zo^m}\frac{1}{2^m}\Pr\left[ \mathcal{H}_k\in R_k, \left|F(\mathcal{H})-w\right|\le \gap m \Cond G^m_{-k}\in R_{-k},F(\mathcal{H})=v  \right]\\
			=& \Pr\left[ G^m_{-k}\in R_{-k} \right]
			\cdot \sum_{v\in\zo^m}\frac{1}{2^m}\Pr\left[ v+ G^m_k\in R_k, \left|v-w\right|\le \gap m \Cond G^m_{-k}\in R_{-k}  \right]\\
			=& \Pr\left[ G^m_{-k}\in R_{-k} \right]
			\cdot \sum_{v\in\zo^m: \left|v-w\right|\le \gap m }\frac{1}{2^m} \sum_{u\in R_k}\Pr\left[ v+ G^m_k = u \Cond G^m_{-k}\in R_{-k}  \right]\\
			=& \frac{1}{2^m}  \Pr\left[ G^m_{-k}\in R_{-k} \right]
			\cdot \sum_{u\in R_k} \sum_{v\in\zo^m: \left|v-w\right|\le \gap m }\Pr\left[ v+ G^m_k = u \Cond G^m_{-k}\in R_{-k}  \right]\\
			<& \frac{1}{2^m}  \Pr\left[ G^m_{-k}\in R_{-k} \right]
			\cdot \sum_{u\in R_k} \frac{1}{2}\sum_{v\in\zo^m}\Pr\left[ v+ G^m_k = u \Cond G^m_{-k}\in R_{-k}  \right]\\
			=&\frac{1}{2} \Pr\left[ \mathcal{H}\in R\right]
		\end{align*}

		Thus we complete the proof of Claim~\ref{claim:ball}	
	\end{proof}

	% \newpage

	\begin{proof}[Proof of Lemma~\ref{lem: partial uniform}]
		We prove this lemma inductively for the indices in $L$.
		In what follows, let $\error_{i}\eqdef \frac{i}{p}$ for every $i\in [\ell]$.
		Given that $\left(G_k^{(1)},\dots,G_k^{(\ell-1)}\right)\Cond G_{-k}^m \in R_{-k}$ is $\error_{\ell-1}$-close to the uniform distribution over $\Z_p^{\ell-1}$, we will show that there exists another instance which, w.l.o.g., we label as $G_k^{(\ell)} $, for which 
		$\left(G_k^{(1)},\dots,G_k^{(\ell-1)},G_k^{(\ell)}\right)\Cond G_{-k}^m \in R_{-k}$ is $\error_{\ell}$-close to uniform distribution over $\Z_p^{\ell}$.
		
		The base case for $\ell=0$ is trivial. 
		In what follows we suppose that the conditional distribution $\left(G_k^{(1)},\dots,G_k^{(\ell-1)}\right)\Cond G_{-k}^m \in R_{-k}$ is already $\error_{\ell-1}$-uniform and we do our induction for $G_k^{(\ell)}$.

		First we fix $x\in \Z_p^{(\ell-1)\times (k-1)}$ for which $\Pr\left[\left(G_{-k}^{(1)},\dots,G_{-k}^{(\ell-1)}\right)=x \cond G_{-k}^m\in R_{-k} \right] \ge \frac{1}{\eta  p^{(\ell-1)(k-1)}}$, 
		and let $\event_x$ denote  the event 
		$\left(G_{-k}^{(1)},\dots,G_{-k}^{(\ell-1)}\right)=x$.  
		Then we discuss the conditional distribution of the remaining instances given $\event_x$.

		%\ge (\frac{1}{p^{\ell-1}}-\error_{\ell-1})\Pr[G_{-k}^m\in R_{-k}]> (\frac{1}{p^{\ell-1}}-\error_{\ell-1})\frac{1}{\alpha k^{\gamma mk}} 

		Let $J_x \eqdef \set{j\in [k-1] \; \Big| \; \Pr\left[G^m_j\in R_j\; \big| \; \event_x\right] \ge {1}/{\beta k^{2\gamma m}} }$,
		\\
		Then
		\begin{align}\label{ineq:A1_upper}
			%\frac{1}{\eta  p^{(\ell-1)(k-1)}} < 
			\Pr\left[ G^m_{-k}\in R_{-k} \cond \event_x \right] =
		\prod_{j\in [k-1]} 
			\Pr\left[ G_j^m\in R_j \cond  \event_x \right]
			\le 
			\prod_{j\in ([k-1]\backslash J_x)} \frac{1}{\beta k^{2\gamma m}} = \left( \frac{1}{\beta k^{2\gamma m}} \right)^{k-1-|J_x|}
		\end{align}

		On the other hand, recalling that $\left(G_{-k}^{(1)},\dots,G_{-k}^{(\ell-1)}\right)$ is uniformly distributed and hence $\Pr[\event_x]=\frac{1}{p^{(\ell-1)(k-1)}}$, we have
		\begin{align}
			& \Pr\left[ G^m_{-k}\in R_{-k} \cond \event_x \right] \notag \\
			=
			&\Pr\left[G^m_{-k}\in R_{-k},\event_x\right] / \Pr[\event_x]
			\notag\\
			= & \Pr\left[ \event_x \Cond G_{-k}^m\in R_{-k} \right] \cdot \Pr[G_{-k}^m\in R_{-k}]/\Pr[\event_x] \notag \\
			\ge & \frac{1}{\eta p^{(\ell-1)(k-1)}} \cdot \frac{1}{\alpha k^{\gamma mk}} \Big/ \left(\frac{1}{p^{(\ell-1)(k-1)}}\right)
			=\frac{1}{\eta \alpha k^{\gamma mk}}
			\label{ineq:A1_lower}
		\end{align}

		Combining \cref{ineq:A1_upper,ineq:A1_lower} 
		and letting $\beta\ge (\eta\alpha)^{{2\gamma}/{\gamma k}}$, we can conclude
		\[\left( \frac{1}{\beta k^{2\gamma m}} \right)^{k-1-|J_x|}\ge  \frac{1}{\eta\alpha k^{\gamma mk}} \implies |J_x|\ge k-1-\frac{\gamma mk\log k+ \log \eta\alpha}{2\gamma m\log k + \log \beta} \ge \left(1-\frac{\gamma}{2\gamma}\right)k-1\]

		Thus the size of $J_x$ is at least
		 $|J_x|\ge \left(1-\frac{\gamma}{2\gamma}\right)k-1=\bigomega{k}$.

		For every $j\in J_x$, we have $\Pr\left[G^m_j\in R_j\; \big| \; \event_x\right] \ge {1}/{\beta k^{2\gamma m}}$ by definition of $J_x$
		and hence 
		\begin{align}\label{ineq: entropy of Gmj}
		\entropy\left[ G^m_j \cond G^m_j\in R_j,\event_x \right] \ge \log\left( {p^{m-\ell}}/{\beta k^{2\gamma m}} \right)=(m-\ell)\log p - 2\gamma m\log k - \log \beta
		\end{align}

		Note that for every $i\in[m]$, $G_j^{(i)}$ is uniform over $\Z_p$ as long as $j\in [k-1]$.
		Thus conditioned on $\event_x$ and $G_j^m\in R_j$,
		if $\exists a\in\Z_p$, 
		$\Pr[G_j^{(i)}=a \cond G_j^m\in R_j,\event_x ]=p_a>\frac{1}{2}$ 
		then we have an upper bound for the conditional entropy of $G_j^{(i)} $: 
		\begin{align*}%\label{ineq: entropy of Gij}
			\entropy[G_j^{(i)}\cond G_j^m\in R_j,\event_x ]\le p_a\log\frac{1}{p_a}+(1-p_a)\log(p-1)< (1+\log(p-1))/2
		\end{align*}

		Let $I_{j,x}$ be defined as 
		\[I_{j,x}\eqdef\set{i\in [m] \st \Hmin\left[ G_j^{(i)}\cond G_j^m\in R_j,\event_x  \right]\ge 1 } = \set{i \st \forall a, \Pr\left[ G_j^{(i)}=a\cond G_j^m\in R_j,\event_x  \right]\le \frac{1}{2} }\]
		Then $\forall i\in \overline{I_{j,x}}\eqdef \left(([m]\backslash[\ell-1]) \backslash I_{j,x}\right)$, $\entropy[G_j^{(i)}\cond G_j^m\in R_j,\event_x ] < (1+\log(p-1))/2$,
		and in particular for $i\in [\ell-1]$, $\entropy[G_j^{(i)}\cond G_j^m\in R_j,\event_x ]=0$ since $G_j^{(i)}$ is already fixed in $\event_x$.
		
		\begin{align*}
			&\entropy\left[ G^m_j \cond G^m_j\in R_j,\event_x \right]
			\le  \sum_{i=1}^{m} \entropy[G_j^{(i)}\cond G_j^m\in R_j,\event_x ]\\
			&=  \sum_{i\in I_{j,x} } \entropy[G_j^{(i)}\cond G_j^m\in R_j,\event_x ] + \sum_{i \notin I_{j,x} \cup \overline{I_{j,x}}} \entropy[G_j^{(i)}\cond G_j^m\in R_j,\event_x ] + \sum_{i\in \overline{I_{j,x}}} \entropy[G_j^{(i)}\cond G_j^m\in R_j,\event_x ] \\
			&= \sum_{i\in I_{j,x} } \entropy[G_j^{(i)}\cond G_j^m\in R_j,\event_x ] + 0 + \sum_{i\in \overline{I_{j,x}}} \entropy[G_j^{(i)}\cond G_j^m\in R_j,\event_x ] \\
			&\le  |I_{j,x}| \cdot \log p + (m-\ell+1 - |I_{j,x}|)(1+\log(p-1))/2
		\end{align*}
		Combining the above with  the lower bound for $\entropy\left[ G^m_j \cond G^m_j\in R_j,\event_x \right]$ in (\ref{ineq: entropy of Gmj}), 
		\begin{align*}
			 (m-\ell)\log p - 2\gamma m\log k - \log \beta & \le |I_{j,x}| \cdot \log p + (m-\ell+1 - |I_{j,x}|)(1+\log(p))/2\\
			\implies \left( \frac{\log p -1}{2} \right) |I_{j,x}| & \ge  (m-\ell)\left(\frac{\log p -1}{2}\right)-2\gamma m\log k -\frac{1+\log p}{2} -\log \beta
		\end{align*}

		Therefore, recalling that $p> k^{1/4}$ 
		{and for $\log \beta =o(\log p)=o(\log k)$}, we have
		\[|I_{j,x}| \ge  m-\ell -\frac{4\gamma m \log k}{\log p-1} - \bigo{\frac{\log \beta}{\log p}}>m-\ell -\frac{4\gamma m \log k}{\frac{1}{4}\log k-1} - \littleo{1} >m-\ell - 18\gamma m + 1 \]

		Therefore, for every $x\in \Z_p^{(\ell-1)\times(k-1)}$ for which 
		\begin{center}$\Pr\left[\left(G_{-k}^{(1)},\dots,G_{-k}^{(\ell-1)}\right)=x \cond G_{-k}^m\in R_{-k} \right] \ge \frac{1}{\eta  p^{(\ell-1)(k-1)}},$\end{center} 
		the size of $|J_x|\ge \left(1-\frac{\gamma}{2\gamma}\right)k-1=\bigomega{k}$; and for every  $j\in J_x$, 
		$|I_{j,x}|>m-\ell - 18\gamma m + 1$ and $\left|\overline{I_{j,x}}\right|=m-\ell+1 - |I_{j,x}| <18\gamma m$.
		
		That is, these three bounds hold with probability 
		at least $1-\frac{1}{\eta}$ by taking a union bound over all $x \in Z_p^{(\ell-1)\times(k-1)}$ where
		$$\Pr\left[\left(G_{-k}^{(1)},\dots,G_{-k}^{(\ell-1)}\right)=x \cond G_{-k}^m\in R_{-k} \right] < \frac{1}{\eta  p^{(\ell-1)(k-1)}}$$
		for 
		% $x$ following the conditional distribution 
		$x\sim\left(G_{-k}^{(1)},\dots,G_{-k}^{(\ell-1)}\right)\Cond G_{-k}^m\in R_{-k}$.
		In what follows we abuse notation a little by assuming $\mathcal{X}\eqdef \left(G_{-k}^{(1)},\dots,G_{-k}^{(\ell-1)}\right)\Cond G_{-k}^m\in R_{-k}$ is a distribution over $\Z_p^{(\ell-1)(k-1)}$ for which $\mathcal{X}$ satisfies all the above statements of $J_x$ and $I_{j,x}$.
		This causes at most an additional loss of $\frac{1}{\eta}$ in the error probability.

		Notice that 
		the conditional distribution $\left(G_{-k}^{(1)},\dots,G_{-k}^{(\ell-1)}\right)\Cond G_{-k}^m\in R_{-k}$ is indeed a product distribution since $R$ is a rectangle.
		That is, letting $x=(x_1,\dots, x_{k-1})$ where $x_j\in \Z_p^{\ell-1}$ for $j\in [k-1]$,
		then $\event_x$ can be decomposed into $k-1$ independent events $\event_{x_j}$, 
		where each $\event_{x_j}$ denotes the event $\left(G_{j}^{(1)},\dots,G_{j}^{(\ell-1)}\right)=x_j$ and $\event_x=\wedge_{j=1}^{k-1}\event_{x_j}$. 
		Therefore the conditional distribution
		$G^m_j \; \big| \; \event_x$ is identical to $G_j^m \Cond  \event_{x_j}$ since the distribution of $G^m_j$ is independent from inputs of the remaining $k-2$ players (among the first $k-1$ players) in the product distribution. 
		As a result, we have
		$\Pr\left[G^m_j\in R_j\; \big| \; \event_x\right] = \Pr\left[ G_j^m\in R_j \Cond  \event_{x_j} \right]$
		so that $\event_{x_j}$ and $x_j$ fully determines whether $j\in J_x$ following the definition of $J_x$.
		Similarly we have $G_j^{(i)}\Cond \set{G_j^m\in R_j,\event_x}$ identical to $G_j^{(i)}\Cond \set{G_j^m\in R_j,\event_{x_j}}$, 
		so that $I_{j,x}$ is also fully determined by $x_j$ and $\event_{x_j}$.

		Next we fix $j\in [k-1]$ and pick $x_j\in\Z_p^{\ell-1}$ for which 

		$j\in J_{x}$ for $x$ extended from $x_j$.
		Now we have $\event_{x_j}$ and $I_{j,x_j} \eqdef I_{j,x}$ containing all but a fraction of $<\frac{18\gamma m}{m-\ell+1}$ coordinates, since $\left|\overline{I_{j,x_j}}\right|< 18\gamma m$ out of the $m-\ell+1$ unfixed coordinates in total.
		Then for $X_j\follows \mathcal{U}_{\Z_p^{\ell-1}}$ and $\mathcal{I}(\cdot)$ denoting the indicator function, 
		\begin{align*}
			& \sum_{i=\ell}^m \mathcal{I}\left( \Pr_{X_j}\left[ i\in \overline{I_{j,X_j}} \Cond j\in J_{X_j} \right] \ge \frac{1}{3}\right)\\
			\le & \sum_{i=\ell}^m 3 \Pr_{X_j} \left[ i\in \overline{I_{j,X_j}} \Cond j\in J_{X_j} \right] \\
			= & 3\sum_{i=\ell}^m \sum_{x_j\in \Z_p^{\ell-1}: j\in J_{x_j}} \Pr[X_j={x_j}] \cdot \mathcal{I}\left(  i\in \overline{I_{j,x_j}} \right)\\
			= & 3 \sum_{x_j\in \Z_p^{\ell-1}: j\in J_{x_j}} \Pr[X_j={x_j}\cond j\in J_{X_j}] \cdot \sum_{i=\ell}^m  \mathcal{I}\left(  i\in \overline{I_{j,x_j}} \right)\\
			= &  3 \sum_{x_j\in \Z_p^{\ell-1}: j\in J_{x_j}} \Pr[X_j={x_j}\cond j\in J_{X_j}] \cdot \left| \overline{I_{j,x_j}} \right|
			< 54 \gamma m
		\end{align*}

		That is, for every fixed  $j\in [k-1]$, there are at least $m-\ell+1-54\gamma m$ coordinates $i\in [m]$ satisfying $\Pr \left[ i\in I_{j,X_j} \Cond j\in J_{X_j} \right] > \frac{2}{3}$,
		i.e., with probability $\frac{2}{3}$, $G_j^{(i)}$ satisfies $\Hmin\left[ G_j^{(i)}\cond G_j^m\in R_j,\event_x  \right]\ge 1$ for a randomly selected $x_j$ conditioned on that $j\in J_{x_j}$ specifies a big component in the rectangle.
		This is exactly the probability that the $i$-th coordinate $G_j^{(i)} $ of $G_j^m$ can be decomposed into a convex combination of a uniform distribution over $2$ elements.

		Now we have at least $(m-\ell+1-54\gamma m)(k-1)$ pairs of $(i,j)\in \set{\ell,\ell+1,\dots,m}\times[k-1]$ satisfying the above condition $\Pr \left[ i\in I_{j,X_j} \Cond j\in J_{X_j} \right] > \frac{2}{3}$, 
		which means at least one fixed $i$ must appear in 
		$\frac {(m-\ell+1-54\gamma m)(k-1)}{m-\ell+1} = \left(1- \frac {54\gamma m}{m-\ell+1}\right)(k-1)$ many pairs for different $j\in[k-1]$ by a standard averaging argument.
		Without loss of generality we may assume $i=\ell$, and let $G''\eqdef(G''_1,\dots,G''_k)$ denote the conditional distribution of $G^{(\ell)}$, i.e., each $G_j''\eqdef G_j^{(\ell)}\Cond\set{G_j^m\in R_j,\event_x}$ denotes the conditional distribution of $G_j^{(\ell)}$. 
		Recalling that $|J_x|\ge \left(1-\frac{\gamma}{2\gamma}\right)k-1$, 
		the number of elements in $|J_x|$ hit by those pairs containing $\ell$ is at least
		$$\left(1-\frac{\gamma}{2\gamma}\right)k-1 +  \left(1- \frac {54\gamma m}{m-\ell+1}\right)(k-1) - (k-1) \ge  \left(1-\frac{\gamma}{2\gamma}- \frac {54\gamma m}{m-\ell+1}\right)k-1=\bigomega{k}$$ 

		We say the pair $(i,j)$ is \emph{good for $x$} if $j\in J_{x}$ and $i\in I_{j,x}$.
		Then recalling that $|J_x|\ge \left(1-\frac{\gamma}{2\gamma}\right)k-1$, the expected number of good $(\ell,j)$ over $x\follows\mathcal{X}$ is lower bounded as follows.
		\begin{align*}
			& \E_x\left[ \#\set{j\in [k-1]\st \text{$(\ell,j)$ is good for $x$}} \right]  \\
			=& \E_x\left[ \sum_{j=1}^{k-1} \mathcal{I}\left(\text{$(\ell,j)$ is good for $x$}\right) \right] \\
			\ge & \E_x\left[ \sum_{j=1}^{k-1} \E_x\left[\mathcal{I}\left(\text{$(\ell,j)$ is good for $x$}\right)\right] \right] 
			= \E_x\left[ \sum_{j=1}^{k-1} \Pr_x[\ell \in I_{j,x}, j\in J_x] \right]\\
			\ge & \E_x\left[ \sum_{j\in J_x} \Pr_x[\ell \in I_{j,x}\cond j\in J_x] \right] \\
			%\ge \frac{2}{3}\E_x\left[ \left|J_x\right|\right] \\
			% \ge & \frac{2}{3}\cdot \left( \left(1-\frac{\gamma}{\gamma'}\right)k-1 \right)
			\ge & \frac{2}{3}\cdot\left(\left(1-\frac{\gamma}{2\gamma}- \frac {54\gamma m}{m-\ell+1}\right)k-1\right)
		\end{align*}
		By a Chernoff bound it implies
		\begin{align*}
			& \Pr_x\left[ \#\set{j\in [k-1]\st \text{$(\ell,j)$ is good for $x$}} \le \frac{1}{3} \left(1-\frac{\gamma}{2\gamma} - \frac {54\gamma m}{m-\ell+1}\right)k \right]\\
		\le &\exp\left( -\bigomega{\left(1-\frac{\gamma}{2\gamma}- \frac {54\gamma m}{m-\ell+1}\right)k} \right)
		\end{align*}
		Let $\error' = \exp\left( -\bigomega{\left(1-\frac{\gamma}{2\gamma}- \frac {54\gamma m}{m-\ell+1}\right)k} \right)$ be an upper bound of this error probability.
		Then with probability at least $1-\error'$, the conditional distribution %$G^{(\ell)}_j$ 
		$G_j''$ can be decomposed into a convex combination of uniform distributions over two distinct elements for at least $\frac{1}{3} \left(1-\frac{\gamma}{2\gamma}- \frac {54\gamma m}{m-\ell+1}\right)k$ indices $j\in [k-1]$.

		Next we show that conditioned on the above decomposition, 
		%the distribution of $x$ that there are at least  $\frac{1}{3} \left(1-\frac{\gamma}{\gamma'}- \frac {27\gamma' m}{m-\ell+1}\right)k$ many good pairs, 
		which happens with probability $\ge 1-\error'$,
		the conditional distribution $G''_k$ is close to uniform by the following claim.
		\begin{claim}[Claim 31 in \cite{Viola:2013we}]
			Let $p$ be a prime number. 
			Let $X$ be the sum of $t$ independent random variables each uniform over $\set{a_i,b_i}\subset\Z_p$ for $a_i\ne b_i$.
			Then $X$ modulo $p$ is $\error\le 0.5\sqrt{p}\exp\left(-\bigomega{t/p^2} \right)$ close to uniform.
		\end{claim}

		Plugging our parameters into the above claim and following exactly the same argument as in \cite{Viola:2013we} ($G_k''$ is $\error''$-close to uniform if every component in the above convex decomposition of $G_k''$ is $\error''$-close to uniform),
		the statistical distance between $G_k''=-\sum_{j=1}^{k-1} G_j''$ and the uniform distribution over $\Z_p$ is bounded by
		\begin{align*}
			\error'' \le & 0.5\sqrt{p}\exp\left(-\bigomega{\frac{1}{3} \left(1-\frac{\gamma}{2\gamma}- \frac {54\gamma m}{m-\ell+1}\right)k/p^2} \right)\\
			= & \exp\left(-\bigomega{\left(1-\frac{\gamma}{2\gamma}- \frac {54\gamma m}{m-\ell+1}\right)\sqrt{k}} \right)
		\end{align*}

		Putting it all together, we conclude that $G_k^{(\ell)}\Cond\set{G_{-k}^m\in R_{-k}, G^{(1)},\dots,G^{(\ell-1)}}$ is close to uniform,
		which implies 
		$\left(G_k^{(1)},\dots,G_k^{(\ell-1)},G_k^{(\ell)}\right)\Cond G_{-k}^m \in R_{-k}$ is also close to uniform.
		Moreover, its statistical distance to uniform is bounded by 
		\begin{align*}
			\error_{\ell} \le \error_{\ell-1}+ \frac{1}{\eta}+\error'+\error''
		\end{align*}

		Let $\eta=2p$ and $\beta=2\ge (\eta\alpha)^{2\gamma/(\gamma k)}=2^{\bigo{\frac{\log k}{k}}}$ for $\alpha=p^\bigo{1}$.
		Then for sufficiently large $k$m 
		the above induction argument goes through for $\ell\le (1-135\gamma) m$, with error $\error',\error''$ bounded by
		$$ \error'= \exp\left(-\bigomega{{k}}\right), \error''=\exp\left(-\bigomega{\sqrt{k}}\right) \Longleftarrow 1-\frac{\gamma}{2\gamma}- \frac {54\gamma m}{m-\ell+1}
		\ge 0.1 \Longleftrightarrow \ell \le (1-135\gamma) m +1 $$
		Therefore the conditional distribution $\left(G_k^{(1)},\dots,G_k^{(\ell-1)},G_k^{(\ell)}\right)\Cond G_{-k}^m \in R_{-k}$ is $\error_{\ell}$-close to uniform for $\error_{\ell}$ bounded by $\frac{\ell}{p}$ as follows:
		$$\error_{\ell} \le \error_{\ell-1}+\frac{1}{\eta}+\error'+\error''\le \frac{\ell-1}{p}+\frac{1}{2p}+ \exp\left(-\bigomega{\sqrt{k}}\right)\le \frac{\ell}{p}$$

		Thus we have proved the induction hypothesis for every $\ell \le (1-135\gamma) m$.
		Let $L$ be the first $(1-135\gamma) m$ indices as in the induction hypothesis, 
		we complete the proof of Lemma~\ref{lem: partial uniform} for $|L|= (1-135\gamma) m$ and statistical distance $\frac{|L|}{p}$.
	\end{proof}

\section{Lower bound for Hamming Norm Estimation}
\label{sec:lower bound L0}

In this section we present a space lower bound for single-pass streaming algorithms for $(1\pm \gap)$-approximating the Hamming norm $L_0$ in the strict turnstile model,  which is denoted by $\hest_{\gap}$ as in \cref{para:Hamming estimation}.

Formally, in the Hamming norm estimation problem there is an underlying vector $(x_1,\dots, x_N)$ which starts from the all zero vector and processes up to $m$ updates each of the form $(i,v)\in [N]\times [\pm M]$.
The update $(i,v)$ means one should add $v$ to the $i$-th coordinate $x_i$ in the vector $x$. 
After processing all $m$ updates, we have $\norm{x}_0 =|\set{i\st x_i\ne 0}|$ and we want to output a number within $(1\pm\gap)\|x\|_0$ with probability $\ge 2/3$. We additionally assume all players have access to a heavy hitters oracle, which tells them whether the frequency of a given coordinate is greater than $T$. This is a generalization of the case without a heavy hitters oracle, where we simply let $T = mM$ and we know that all frequencies are guaranteed to be smaller.
The strict turnstile model guarantees that $x_i \geq 0$ for all $i \in [N]$ at all positions in the stream,
in which case it suffices to prove the space lower bound in the simultaneous communication model following the reduction in Theorem 4.1 of \cite{AHLW16}.
Furthermore, it is also guaranteed that for every $i \in [N]$,  $x_i \leq \poly(n)$ at the end of the stream. 
In this setting, the algorithm of \cite{KNW10} approximates $\|x\|_0$ up to a $(1 \pm \gap)$ factor with $\bigo{\gap^{-2}\log(N) \left(\log(1/\gap) + \log \log (T) \right)}$ bits of space\footnote{Indeed, their algorithm stores $\bigo{\gap^{-2}\log N}$ counters modulo primes that are each $\bigo{\log(1/\gap) + \log \log (T)}$ bits in magnitude, and it does not matter how large the values of $x_i$ are at intermediate positions in the stream. }, as long as $\gap > 0$.

We first note that solving distinct elements with a heavy hitters oracle reduces to solving distinct elements given a threshold on the frequency of the coordinates. As such, we will solve the complexity question of space complexity given a threshold $T$ for the frequency.

\begin{theorem}
\label{thm:streaming_bound}
The space complexity of $(1\pm \epsilon)$ approximating $L_0$ with probability at least $2/3$ in a strict turnstile stream with access to a heavy hitters oracle with a threshold of $T > 1$ is $\Omega(\epsilon^{-2} \log n \log \log T)$.
\end{theorem}

We note that the assumption $T > 1$ is necessary for this bound to be well defined. When $T=1$, the heavy hitters oracle tells us exactly whether or not the frequency of a coordinate is $0$ at the end of the stream, so the complexity is $\Theta(\log n)$. This lower bound follows as we need to write down the answer and the upper bound follows as we can directly count the elements with nonzero frequency.

To prove this theorem, we first prove the following lemma:

\begin{lemma}
\label{thm:promise_reduction}
The space complexity of $(1\pm \epsilon)$ approximating $L_0$ with probability at least $2/3$ in a strict turnstile stream with access to a heavy hitters oracle with a threshold of $T > 1$ is at least $RSC^T_{2/3}(T_\epsilon)$.
\end{lemma}
\begin{proof}
Suppose we have an algorithm $A$ which gives us a $(1\pm \epsilon)$ approximation of $L_0$ in a strict turnstile stream with access to a heavy hitters oracle with a threshold of $T$.\\

Now, if we are given an input where the maximum frequency of any element is at most $T$, then we can go through our input and do exactly what $A$ would do for everything other than calls to the heavy hitters oracle. If $A$ would make a call to a heavy hitters oracle, we just treat the answer as $0$ without making this query and proceed as $A$ would.\\

Since we assumed the input has a maximum frequency of $T$, the heavy hitters oracle would return $0$ for every element, so this would give us the same answer as $A$, and by correctness of $A$, it's a $(1\pm \epsilon)$ approximation.\\
\end{proof}

Now, we will state and prove our main theorem:
\begin{theorem}\label{thm:L0}
	For error tolerance $\gap <1/3$ and $\gap=\max\set{\bigomega{ \sqrt{\frac{\log k}{k}}}, \frac{1}{{N}^{0.49}}}$, 
	any single-pass streaming algorithm solving $\hest_\gap$ with probability $\ge 2/3$ in the strict turnstile model
	% which $(1\pm\gap)$-approximates $\norm{x}_0$ 
	must use $\bigomega{\gap^{-2}\log (N)\log\log(T)}$ bits of space.
\end{theorem}

First we introduce $\go$ and $\gose$:

\begin{definition}
In the $c$-\go$_n$ problem, we have two players Alice and Bob. They each have as input a vector in $\{0,1\}^n$ and we wish to compute 

$$c\text{-}\go_n(x,y) = \begin{cases}
1, \quad \left|\left(\sum_{i \in n} \textbf{XOR}(x_i,y_i)\right) - \frac{n}{2}\right| \ge 2c\sqrt{n},\\
0, \quad \left|\left(\sum_{i \in n}  \textbf{XOR}(x_i,y_i)\right) - \frac{n}{2}\right| \le c\sqrt{n},\\
\end{cases}
$$
and otherwise, it can return anything.
\end{definition}

\begin{definition}
In the $c$-\gose$_n$ problem, we have two players Alice and Bob. They each have as input a vector in $\mathbb{Z}^n$ and we similar to above, we wish to compute

$$c\text{-}\gose_n(x,y) = \begin{cases}
1, \quad \left|\left(\sum_{i \in n} \mathbf{1}_{x_i + y_i = 0}\right) - \frac{n}{2}\right| \ge 2c\sqrt{n},\\
0, \quad \left|\left(\sum_{i \in n} \mathbf{1}_{x_i + y_i = 0}\right) - \frac{n}{2}\right| \le c\sqrt{n},\\
\end{cases}
$$
and otherwise, it can return anything. We let $\mathbf{1}_{x_i + y_i = 0}$ denote the indicator function which is $1$ if $x_i + y_i = 0$ and $0$ otherwise.
\end{definition}

Additionally, we will let $\sumequal_{k,\delta}^{m,a}$ denote $m$ independent instances of $\sumequal_{k,\delta}$, and our protocol needs to be able to solve at least $am$ of these instances correctly with probability at least $1-\delta$.

\begin{lemma}
for every $k \in \mathbb{N}$, $0 \le \delta \le 1/2$, and $n \ge c^2/100\epsilon^2=n'$,
$$\rcclinT_{k,\delta}(10\epsilon\sqrt{n}\text{-}\gose_n) \ge \rcclinT_{k,\delta}(c\text{-}\gose_{n'})$$
\label{lemma:gap-ort reduction}
\end{lemma}

\begin{proof}[Proof of \cref{lemma:gap-ort reduction}]
	Given $n'=c^2/100\gap^2$ and 
	an input instance of $c$-$\gose_{n'}$ with underlying \sumequal problems outputting $\mathbf{x}'\in \zo^{n'}$, we create the new input to $10\gap\sqrt{n}$-$\gose_n$ by taking $100\gap^2 n / c^2$ copies of each coordinate, 
	with results of underlying problems being $\mathbf{x}\in\set{\pm 1}^n$. 
	As a result, $\sum_{j=1}^{n}\mathbf{x}_j  =  \frac{100\gap^2 n}{c^2}\cdot \sum_{j=1}^{n'}\mathbf{x}'_j$.

	If $|\sum_{j}\mathbf{x}'_j| \le c\sqrt{n'}$,
	then $|\sum_{j}\mathbf{x}_j| \le 10\gap n$,
	and on the other hand 
	$|\sum_{j}\mathbf{x}'_j| \ge 2c\sqrt{n'}$
	implies $|\sum_{j}\mathbf{x}_j| \ge 20\gap n$.

	Thus, any $k$-player $\error$-error simultaneous communication protocol for $10\gap\sqrt{n}$-$\gose_n$ immediately implies a $k$-player $\error$-error simultaneous communication protocol for $c$-$\gose_{n'}$.
	 % where $n'=c^2/\gap^2=\bigtheta{1/\gap^2}$.
	 
	 Since all we are doing is copying coordinates, this does not change the threshold.
\end{proof}

\begin{theorem}
Given some simultaneous communication protocol $\Pi$ with two players that solves $1$-\gose$_n$ when each \sumequal instance has input drawn from its hard distribution $\mu$, there exists a protocol $\Pi'$ such that $\rcclin_{2,\delta}(\Pi') \le O(\rcclin_{2,\delta}(\Pi))$ which solves $\Omega(n)$ of the individual sum-equal instances with probability at least $\frac{1}{2}+\beta$ for some constant $\beta>0$.
\label{thm:gose complexity}
\end{theorem}

\begin{proof}
Suppose $\Pi$ is a protocol that solves $1$-\gose$_n$. Now, if Alice has input $X = (x_1, x_2, \dots x_n)$ and Bob has input $Y = (y_1, y_2, \dots y_n)$ to $1$-\gose$_n$, we define a corresponding instance of $1$-\go$_n$ where Alice gets input $X' = (x'_1, x'_2, \dots x'_n)$ and Bob gets input $Y' = (y'_1, y'_2, \dots y'_n)$ where $y'_i = 0$ with probability $1$ if $y_i < M/2$, probability $\frac{1}{2}$ if $y_i=M$, and probability $0$ otherwise where $M$ is defined as in the proof of \cref{thm:woodruff}, and $x'_i = 1-y'_i$ iff $x_i+y_i = 0$.\\

When $X,Y \sim \mu$, each $x_i+y_i = 0$ with probability $\frac{1}{2}$ and $y_i$ is equally likely to be $-x_i$ or $M-x_i$ so it is symmetric about $M/2$. Hence, $(X',Y') \sim \{0,1\}^{2n}$. Furthermore, the answer to the $1$-\gose$_n$ instance is the same as the answer to the $1$-\go$_n$ instance by construction. Therefore, we can solve this instance of $1$-\go$_n$ by simply running $\Pi$.\\

Now, if we let $\mathcal{M}$ be the message sent by Alice to Bob in protocol $\Pi$, then

$$I(X' ; \mathcal{M}, Y) \ge IC(1\text{-Gap-Ort}_n) = \Omega(n)$$

since Bob can solve $1$-\go$_n$ where Alice has input $X'$ and Bob has input $Y'$ when he has access to $(\mathcal{M},Y)$ by returning the answer to $1$-\gose$_n$ using protocol $\Pi$ with input $Y$ after being sent the message $\mathcal{M}$.

We now note that $X'$ is $n$ iid uniformly random bits. As such,

$$I(X' ; \mathcal{M}, Y) = \sum_{i=1}^n I(x'_i;\mathcal{M}, Y \mid x'_1, x'_2, \dots x'_{i-1}) \ge \sum_{i=1}^n I(x'_i;\mathcal{M}, Y).$$

Each of these terms is upper bounded by $1$, so in order for the sum to be $\Omega(n)$, there exists some constant $c>0$ such that there are at leasts $cn$ indices $j$ such that $I(x'_j;\mathcal{M}, Y) 
\ge \alpha$ for some constant $\alpha > 0$.\\

Now, let
$$J = \{j \mid I(x'_j;\mathcal{M}, Y) \ge \alpha\}.$$

We claim that the transcript of $\Pi$ must contain the solution to the $j^{th}$ Sum-Equal instance with probability at least $\frac{1}{2} + \beta$ for a constant $\beta > 0$ for each $j$. To see this, we note that Bob has as input $Y$ for $1$-\gose$_n$ so he can compute $y'_j$. Then, we note that

$$H(x'_j \mid \mathcal{M},Y) = H(x'_j) - I(x'_j; \mathcal{M},Y) \le 1-\alpha$$

Since $x'_j \in \{0,1\}$, let $\Pr[x'_j = 0 \mid \mathcal{M},Y] = p$. Then, if $p=0$, the entropy is $0$ so this is satisfied for any $0 < \alpha \le 1$. If $p>0$, we have

$$-(p \log p + (1-p) \log (1-p))\le 1-\alpha.$$

Since this is symmetric about $p=\frac{1}{2}$ and cannot be satisfied by $p = \frac{1}{2}$ since $\alpha > 0$, we assume WLOG that $p < \frac{1}{2}$, in which case the entropy monotonically decreases as $p$ decreases. Now, we claim that we must have $p < \frac{1}{2}-\frac{\alpha}{2}$. It suffices to show that

$$ -\left(\left(\frac{1-\alpha}{2}\right) \log \left(\frac{1-\alpha}{2}\right) + \left(\frac{1+\alpha}{2}\right) \log \left(\frac{1+\alpha}{2}\right)\right) \ge 1-\alpha$$

Simplifying this expression yields the solution
$$0 < \alpha < 1.$$

Thus, for $0 < \alpha < 1$, we must have $p < \frac{1-\alpha}{2}$. If $\alpha=1$, then the entropy is $0$ so we must have $p=0$. Thus, we get that $p \le \frac{1-\alpha}{2}$.\\

By symmetry, we thus have that either
$$\Pr[x'_j = 0 \mid \mathcal{M},Y] \le \frac{1-\alpha}{2}$$
or
$$\Pr[x'_j = 0 \mid \mathcal{M},Y] \ge \frac{1+\alpha}{2}.$$
In the former case, Bob lets $\hat{x'_j} = 1$ and in the latter case, Bob lets $\hat{x'_j} = 0$. Bob then computes $y'_j$ from $y_j$. Then, if $\hat{x'_j} = y'_j$, Bob concludes that $x_j + y_j \ne 0$ and if $\hat{x'_j} = 1-y'_j$, Bob concludes that $x_j + y_j = 0$. By construction, this succeeds with probability at least $\frac{1+\alpha}{2}$, and all we did was run $\Pi$ and compute the value from the transcript.
\end{proof}

\begin{corollary}
When $\alpha$ and $c$ are the constants from the proof of $\cref{thm:gose complexity}$,
$$\DlinT_{2,\delta,\mu'}\left(1\text{-}\gose_{\epsilon^{-2}/100}\right) \ge \DlinT_{2,(1+\alpha)/2,\mu}\left(\sumequal^{\epsilon^{-2}/100,c}_2\right)$$
\end{corollary}

\begin{proof}
This follows directly from \cref{thm:gose complexity}. Each instance of Sum-Equal corresponds to a single coordinate from $1$-\gose so their frequencies must all be bounded by $T$ as well.
\end{proof}

\begin{theorem}
When $\delta < \frac{1}{2}$ and $a$ is some constant fraction,
\begin{equation}
    \ic^T_{k,\delta}(\sumequal_k^{n',a}) \ge \Omega(n'k \log \log T)
    \label{eq:Sum-Equal-Complexity}
\end{equation}
where $\sumequal_k^{n',a}$ is the problem where we are given $n'$ independent instances of $\sumequal$ and we are asked to solve $an'$ of them with probability $1-\delta$ each.
\label{thm:woodruff}
\end{theorem}
The proof of this theorem can be found in \cref{sec:over Z}.

\begin{corollary}
For the input distribution $\mu$ defined in the proof of \cref{thm:woodruff}, $\delta < \frac{1}{2}$, and $0<a<1$,
$$\DlinT_{2,\delta,\mu}(\sumequal_2^{n',a}) \ge \Omega(n' \log \log T)$$
\end{corollary}
\begin{proof}
If we plug $k=2$ into (\ref{eq:Sum-Equal-Complexity}), we get

$$\DlinT_{2,\delta,\mu}(\sumequal_2^{n',a}) \ge \ic^T_{2,\delta}(\sumequal_2^{n',a}) \ge \Omega(n' \log \log T)$$

since by definition $\mu$ is the hard distribution that we got the information complexity bound from.

\end{proof}

\begin{definition}
The Aug-Index-GOSE$_{n,k}^t$ problem consists of $t$ independent instances of $\epsilon\sqrt{n}$-Gap-Ort-Sum-Equal$_n$, denoted $g_1, g_2, \dots g_t$, with $k$ players and $n$ coordinates each. In this problem, the referee is asked to estimate $g_i$ based on an index $i\in[t]$ together with the auxiliary information of $f_{i+1},\dots, f_t$, where for convenience we let $f_i \in [\pm n]$ denote the bias of the number of underlying $\sumequal_k$ instances outputting $1$ in $g_i$.
\end{definition}

\begin{theorem}
$\rcclinT_{k,1/3}(T_\epsilon) \ge \rcclinT_{k,0.4}(\auggo)$
\end{theorem}

\begin{proof}
We let the $i$-th $\text{$\gap\sqrt{n}$-$\gose_n$}$ instance $g_i$ in the $\auggo_{n,k}^{t}$ problem have frequency $100^{i-1}$, i.e., each element in $g_i$ is counted $100^{i-1}$ times (as that many distinct elements).
	Thus the universe contains $N\eqdef n+100\cdot n+\dots+100^{t-1}\cdot n \le 100^t n/99$ distinct elements in total, which is $N\le n^{1.01}$ for sufficiently small $t$ (and hence $1/N^{0.49}>1/\sqrt{n}$).
	The final Hamming norm is a weighted sum $F'\eqdef \sum_{i=1}^t 100^{i-1} f'_i$.
	The advantage of $F'$ is hence $F \eqdef 2F'-N = \sum_{i=1}^t 100^{i-1} f_i$. 

	Then we invoke the simultaneous communication protocol for $\hest_\gap$ to estimate $F'$, which returns a value $\widetilde{F'}$ satisfying 
	$(1-\gap)F' \le \widetilde{F'} \le (1+\gap)F'$.
	Translating to the advantage we get $\left| \widetilde{F} -  F \right| \le 2\gap F' \le 2\gap N$.
	From this approximated value $\widetilde{F}$, together with the index $i$ and auxiliary information $f_{i+1},\dots, f_t$, we need to determine the output value of $g_i$.
	Since the influence of $f_{j}$ with $j>i$ can be precisely removed from $F$ before getting the approximated norm $\widetilde{F}$, in what follows it suffices to consider the estimation of $g_t$ when the index is indeed $i=t$.
	Recall that $F = 100^{t-1} f_t + \sum_{i=1}^{t-1} 100^{i-1} f_i$,
	and thus $\widetilde{F}$ is also an approximation of $100^{t-1} f_t$ as long as the additive error $\sum_{i=1}^{t-1} 100^{i-1} f_i$ is bounded.

	Let the input distribution to every $f_i$ be padded from the \text{$1$-$\gose_{\epsilon^{-2}}$} distribution $\mu'$ as in \cref{thm:gose complexity}, where the coordinates are iid bits drawn uniformly from $\{0,1\}$. Thus, each $f_i$ has expectation $0$ and variance $25\epsilon^2n^2$
	It immediately follows by Chebyshev's inequality that $\Pr\left[|f_i|\ge 50 \gap n\right] \le 1/100$.
	Similarly, $\Pr\left[|f_{i}|\ge 50^j \gap n\right] \le 1/100^j$.
	Therefore,
	\begin{align}
		\Pr\left[ \left|\sum_{i=1}^{t-1} 100^{i-1} f_i \right|>  100^{t-1}\gap n\right] 
		& \le \sum_{i=1}^{t-1} \Pr\left[ \left|f_{t-i}\right| >  50^i \gap n \right]
		\le \sum_{i=1}^{t-1} \frac{1}{100^i} \le \frac{1}{99}
	\end{align}
	where the first inequality holds because if $\left|f_{t-i}\right| \le 50^i \gap n $ for every $i$, then $ \left|\sum_{i=1}^{t-1} \cdot 100^{i-1} f_i \right|\le \sum_{i=1}^{t-1} 100^{i-1}\times 50^{t-i} \gap n\le \frac{50^t}{100}\gap n\sum_{i=1}^{t-1} 2^i < 100^{t-1}\gap n$.

	Notice that as long as $\widetilde{F}$ is a $(1\pm \gap)$-approximation of $F$, we must have $\left| \widetilde{F} -  F \right| \le 2\gap N$. Furthermore suppose that we return $0$ if $\widetilde{F} < 15 \cdot 100^{t-1}\epsilon n$ and $1$ if $\widetilde{F} \ge 15 \cdot 100^{t-1}\epsilon n$. Since we know that $N \le 100^{t}n/99$, we have
	
	$$2 \gap N \le 2 \epsilon 100^{t}n/99 < 3 \cdot 100^{t-1} \epsilon n.$$
	
	So in particular, if $\hest_\gap$ succeeds, if $g_t = 0$, we have $|f_t| \le 10 \cdot 100^{t-1}\epsilon n$, so $|F| = |\sum_{i=1}^{t}100^{i-1}f_i| \le 11 \cdot 100^{t-1} \epsilon n$ with probability at least $\frac{98}{99}$. Then, $|\tilde{F}| < 14 \cdot 100^{t-1} \epsilon n$ and our algorithm succeeds.\\
	
	Similarly, if $g_t = 1$, we have $|f_t| \ge 20 \cdot 100^{t-1} \epsilon n$. Thus, with probability at least $\frac{98}{99}$, $|F| = |\sum_{i=1}^{t}100^{i-1}f_i| \ge 19 \cdot 100^{t-1} \epsilon n$, so $|\tilde{F}| > 16 \cdot 100^{t-1} \epsilon n$ and our algorithm succeeds.\\
	
	Thus, if $\hest_\gap$ succeeds with probability $\frac{2}{3}$, the above algorithm succeeds with probability $\frac{2}{3} \cdot \frac{98}{99} > 0.6$. Thus we can determine the value of $g_t$ with probability $\ge 0.6$. The thresholds stay the same because all we did to change the input was copy coordinates, which does not change the frequencies. Hence,
	\begin{align}\notag
		\rcclinT_{k,1/3}\left(\hest_\gap \right) 
		\ge \rcclinT_{k,0.4}\left( \auggo_{n,k}^{t} \right).
		% \ge \bigomega{t} \cdot \rccsim_{k,0.1}\left(\gap\sqrt{n}-\gs_n\right)
	\end{align}

	Finally, to bound the complexity $\rscT_{k,0.4}( T_\epsilon)$, we conclude as follows:\\
	
	When we solve $\auggo$, we claim that in order to solve $\auggo$, with probability at least $0.6$ on every input, we must solve every instance of $\gose$ with probability at least $0.6$.
	
	To see this, suppose we have a protocol that solves $\auggo$ with probability at least $0.6$. let the index corresponding to this index be $j$. Then, we can choose our input such that the index is $j$, and the probability of success of our protocol is the same as the probability that the protocol solves instance $j$ correctly. By assumption, our protocol succeeds with probability at least $0.6$, so it solves the $j^{th}$ instance of $\gose$ with probability at least $0.6$. This holds for every $j$, so our protocol must solve every instance of $\gose$ with probability at least $0.6$.

	Thus,
	
	\begin{align*}
	\rcclinT_{k,2/3}(T_\epsilon) &\ge \rcclinT_{k,0.4}\left( \auggo_{n,k}^{t} \right)\\
	&\ge \rcclinT_{k,0.4}\left(10\epsilon\sqrt{n}\text{-}\gose_{n,k}^t\right)\\
	&\ge \rcclinT_{k,0.4}\left(1\text{-}\gose_{n',k}^t\right)\\
	&\ge \DlinT_{k,\delta,\mu}\left(1\text{-}\gose_{n',k}^t\right)\\
	&\ge (k-1) \DlinT_{2,\delta,\mu}\left(1\text{-}\gose_{n',2}^t\right)\\
	&\ge (k-1) \DlinT_{k,\delta,\mu}\left(\sumequal_{2}^{tn',c}\right)\\
	&\ge (k-1) \ic^T_{k,\delta}\left(\sumequal_2^{tn',c}\right)\\
	&\ge \Omega(ktn' \log \log T) = \Omega(\epsilon^{-2} k\log n \log \log T)
	\end{align*}
	so
	$$\rscT_{k,0.4}(T_\epsilon) \ge \frac{1}{k} \rcclinT_{k,0.4}(T_\epsilon) \ge \Omega(\epsilon^{-2} \log n \log \log T)$$
	
\end{proof}

\appendix

% !TEX root = ./communication_separation.tex

\section{Communication Upper Bound for \textsc{Equality}}
\label{sec:testing equality}

The standard $\error$-error protocol solving the $\textsc{Equality}$ problem starts by sending and comparing the digest under a random hash function $h:[p]\to [q]$ where $q=\bigo{\error^{-1}\log p}$. 
For example, let $q$ be a random prime drawn from the interval $[\error^{-2}\log^2 p, 2\error^{-2}\log^2 p]$ and let $h$ compute a number modulo $q$.
By the prime number theorem there are at least $2\sqrt{N}$ primes in the interval $[N,2N]$, 
which implies the existence of $2\error^{-1}\log(p)$ distinct primes in that range. 
% $[\error^{-2}\log^2 p, 2\error^{-2}\log^2 p]$.
For any two distinct numbers $x,y\in \Z_p$,  since $z=x-y$ has no more than $\log|z|\le \log p$ prime factors, the error probability of the protocol is bounded by the collision probability of $h$ as follows:
%$h(x)-h(y)=(x\mod q)-(y\mod q)$, and hence
\[
\Pr_q\left[h(x)= h(y) \right]=\Pr_q\left[x\equiv y \pmod q \right]=\Pr_q\left[q | (x-y) \right] \le \frac{\log p}{2\error^{-1}\log p} <\error
\]
% =\frac{\error}{2}$$

The communication is a message of the form $(h,h(x))$ (indeed $(q, x\mod q)$ in the above example), whose length is at most $2\ceil{\log q}=\bigo{\log(1/\error)+\log\log p}=\bigo{\log(1/\error)+\log\log k}$ bits.
In particular this is an upper bound for one-way communication protocols computing $\textsc{Equality}$.
Recalling that $p=\bigtheta{k^{1/4}}$,
 % in Lemma~\ref{lemma:log/loglog-separation},
we can conclude
\[
\rcc_{2,\error} (f) \le\rccone_{2,\error}(f)=\bigo{\log(1/\error)+\log\log k}
\]

We note that the $1/\error$ factor in $q$ is unavoidable, since otherwise more than an $\error$ fraction of numbers would share the same message and hence the collision probability, as well as the error probability, would exceed $\error$.

\section{The lower bound for $\sumequal^{m,a}_k$ over integers}

\label{sec:over Z}

\noindent\textbf{Theorem~\ref{thm:woodruff} (restated).} 
{\it
	Let $\Pi$ be the $\error$-error simultaneous $k$-player protocol for solving m independent instances of the $\sumequal_k^{m,a'}$ problem, where 
	$m \le\frac{k\log\log T}{20\log k}$ and the error tolerance $\error \in (0,1/6)$.
	The simultaneous communication complexity of $\Pi$ is $\rcclinT_{k,\error}(\Pi)=\bigomega{mk\log\log T}$.
}
\medskip
\begin{proof}
	To prove the $\bigomega{mk\log\log T}$ lower bound we will deduce a contradiction if $\Pi$ uses $c<\gamma mk\log\log T$ bits of communication, for a \emph{sufficiently small} constant $\gamma$.
	By decreasing $\gamma$ we may assume that $k$ is arbitrarily large.

	For the hard distribution we first introduce a magnitude bound $a$ defined to be the largest integer such that $a! \le T$. We define $M = a!$. We note that $M \le T$ and $a = O(\log M)$ so $a = O(\log T)$. Let $a = \gamma' \log T$.

	Now we specify the distribution $\mathcal{H}$ for the $\sumequal_k$ instances.
	$\mathcal{H} \eqdef \left(G/2+B/2\right)^m$ consists of $m$ independent copies of $G/2+B/2$,
	for $G$, $B$ defined as follows:
	$$
		\left\{
		\begin{array}{c c r}
			G \eqdef &	\big( G_1,\dots, G_{k-1} , &-\sum_{j=1}^{k-1} G_j \big)		\\
			B \eqdef &  \big( B_1,\dots, B_{k-1} , &M-\sum_{j=1}^{k-1} B_j \big)
		\end{array}		
		\right.
	$$
	where $G_j,B_j$ are uniformly and independently chosen from $[a]$ for every $j\in[k-1]$. 
	Note that: 
	a) $\sumequal_{k}(G)=0$, $\sumequal_{k}(B)=1$; 
	b) the first $k-1$ elements of $G$ and $B$, denoted by $G_{-k}$ and $B_{-k}$, are the same uniform distribution over $[a]^{k-1}$.
	Thus we can write $B=(G_{-k},M+G_k)$;
	c) for $j\in[k-1]$, the $j$-th player's input $\mathcal{H}_j$ is uniform over $[a]^m$ and independent from other players' input.

	Besides $\mathcal{H}_k$, the referee gets in addition an index $n$ uniformly drawn from $[m]$ together with the answers $Y^{(j)}=\sumequal_k(X^{(j)})$ for $j=n+1,\ldots,m$.
	Let $\mathcal{H}_n'\eqdef (\mathcal{H},Y^{(n+1)},\ldots,Y^{(m)})$
	and the hard input distribution is defined as $\mathcal{H}'\eqdef % (\mathcal{H},n,Y^{(n+1)},\ldots,Y^{(m)})=
	\sum_{n=1}^m \frac{1}{m}\cdot \mathcal{H}_n'$.

	Now we derandomize the protocol $\Pi$ by fixing the randomness and thus get an $\error$-error deterministic protocol $\Pi'$ with respect to the above input distribution.
	That is, $\Pi'$ outputs $\sumequal_k^{(n)}=\sumequal_k(X^{(n)})$ with probability $\ge 1-\error$.

	By averaging, for at least $m/2$ choices of the index $n\in [m]$ and the restricted distribution $\mathcal{H}'_n$,
	the error of $\Pi'$ is bounded by $2\error$.
	\begin{align}
		\Pr_{(X,Y)\follows \mathcal{H}'_n}\left[ \Pi'_{out}(\Pi'(X,Y)) \ne \sumequal_k(X^{(n)}) \right] \le 
	2\error
	\end{align}

	Then we introduce \cref{lemma: lb mutual information} that 
	lower bounds $I\left(X^{(n)}_{-k}; M_1,\ldots, M_{k-1}\right)\ge 0.1{k\log a}$ for protocols with small error.
	For compactness the proof of \cref{lemma: lb mutual information} is deferred to the end of this section.

	\begin{lemma}\label{lemma: lb mutual information}
		For every $n$  such that $\Pi'$ errs with probability $\le 1/3$ on input  $(X,Y)\follows \mathcal{H}'_n$, on at least $a'm$ of the \sumequal instances,
		the mutual information between $X^{(n)}$ and $\Pi'(X,Y)$ must be
		$I\left(X^{(n)}_{-k}; M_1,\ldots, M_{k-1}\right)\ge 0.1{k\log a}$.
	\end{lemma}

	Using \cref{lemma: lb mutual information}, it immediately follows that for $\error \le 1/6$ the protocol $\Pi'$ must use $\bigomega{mk\log a}$ bits of communication.
	Since
	\begin{align*}
	 	\rccsim_{k,\error}(\Pi') & \ge I\left(X_{-k}; M_1,\dots,M_{k-1} \right)\\
	 	& = \sum_{i=1}^m I\left(X^{(i)}_{-k}; M_1,\dots,M_{k-1}\cond X^{(1)}_{-k},\ldots,X^{(i-1)}_{-k} \right)\\
	 	& = \sum_{i=1}^m I\left(X^{(i)}_{-k}; M_1,\dots,M_{k-1}, X^{(1)}_{-k},\ldots,X^{(i-1)}_{-k} \right)\\
	 	& \ge \sum_{i=1}^m I\left(X^{(i)}_{-k}; M_1,\dots,M_{k-1} \right)\\
	 	& \ge \frac{a'm}{2}\cdot 0.1 k\log a =\bigomega{mk\log a}
	\end{align*} 
	since $a'$ is some constant between $0$ and $1$.
\end{proof}	

\begin{proof}[Proof of Lemma~\ref{lemma: lb mutual information}]
		Suppose  by contradiction that
		$I\left(X^{(n)}_{-k}; M_1,\ldots, M_{k-1}\right)<0.1 k\log a$ and recall that $m \le\frac{k\log\log T}{20\log k} \le\frac{0.1 k\log a}{\log (ka)}$ for $a=\gamma'\log T$ and sufficiently large $T$,
		$$ I\left(X^{(n)}_{-k}; M_1,\ldots, M_{k-1}, X_k, Y^{(n+1)},\ldots, Y^{(m)} \right)<0.1 k\log a + {m\log(ka)}<0.2 k\log a$$

		Therefore, recalling that $I(A;B,C)=I(A;B\cond C)$ when $A$ is independent from $C$ and that $X_j^{(n)}$ is independent from $X^{(n)}_1,\ldots, X^{(n)}_{j-1}$, 
		\begin{align*}
			& \sum_{j=1}^{k-1} I\left(X^{(n)}_{j}; M_1,\ldots, M_{k-1}, X_k, Y^{(n+1)},\ldots, Y^{(m)} \right) \\
			\le & \sum_{j=1}^{k-1} I\left(X^{(n)}_{j}; M_1,\ldots, M_{k-1}, X_k, Y^{(n+1)},\ldots, Y^{(m)}, X^{(n)}_1,\ldots, X^{(n)}_{j-1} \right) \\
			= & \sum_{j=1}^{k-1} I\left(X^{(n)}_{j}; M_1,\ldots, M_{k-1}, X_k, Y^{(n+1)},\ldots, Y^{(m)} \cond X^{(n)}_1,\ldots, X^{(n)}_{j-1} \right)\\
			\le & I(X^{(n)}_{-k}; M_1,\ldots, M_{k-1}, X_k, Y^{(n+1)},\ldots, Y^{(m)} ) <0.2 k\log a
		\end{align*}
		
		As a result, there is $J\subseteq [k-1]$ and $|J|>k/2$ such that for every $j\in [k-1]$, it holds that
		$I\left(X^{(n)}_{j}; M_1,\ldots, M_{k-1}, X_k, Y^{(n+1)},\ldots, Y^{(m)} \right)<-1+ 0.5\log a$,
		and hence 
		\begin{align}
			&\entropy\left[X^{(n)}_{j} \cond M_1,\ldots, M_{k-1}, X_k, Y^{(n+1)},\ldots, Y^{(m)} \right] \notag\\
			= &\entropy\left[X^{(n)}_{j}\right] -  I\left(X^{(n)}_{j}; M_1,\ldots, M_{k-1}, X_k, Y^{(n+1)},\ldots, Y^{(m)} \right)\notag\\
			>& \log a - (-1+ 0.5\log a) = 1+ 0.5\log a \label{ineq:lower bound for H}
		\end{align}
	
		Note that
		$\Hmin\left[X^{(n)}_{j} \cond M_1,\ldots, M_{k-1}, X_k, Y^{(n+1)},\ldots, Y^{(m)} \right] <1$ implies the existence of $x\in [a]$ such that  $\Pr\left[X^{(n)}_{j}=x \cond M_1,\ldots, M_{k-1}, X_k, Y^{(n+1)},\ldots, Y^{(m)} \right]=p_x>\frac{1}{2}$,
		and hence it follows that 
		\begin{align}
			\entropy\left[X^{(n)}_{j} \cond M_1,\ldots, M_{k-1}, X_k, Y^{(n+1)},\ldots, Y^{(m)} \right] 
			=& \sum_{i\in[a]} p_i\log\frac{1}{p_i} \notag \\
			\le & p_x\log \frac{1}{p_x}+(1-p_x)\log\frac{a-1}{1-p_x}\notag\\
			< & 1+0.5\log (a-1) \label{ineq:upper bound for H}
		\end{align}

		Thus, (\ref{ineq:lower bound for H}) ensures that $\Hmin\left[X^{(n)}_{j} \cond M_1,\ldots, M_{k-1}, X_k, Y^{(n+1)},\ldots, Y^{(m)} \right]\ge 1$ for every $j\in J$.
		In what follows, we prove that if $\Hmin\left[X^{(n)}_{j} \cond M_1,\ldots, M_{k-1}, X_k, Y^{(n+1)},\ldots, Y^{(m)} \right]\ge 1$ for every $j\in J$ and $|J|>k/2$, 
		then the conditional distribution
		$B_k'\eqdef G_k'+M$  and
		$G_k'\eqdef -\sum_{j=1}^{k-1}X^{(n)}_{j} \Cond \set{M_1,\ldots, M_{k-1}, X_k, Y^{(n+1)},\ldots, Y^{(m)}}$ have statistical distance $\le k^{-1/8}$.

		Notice that for $j\in J$ and $\Hmin\left[X^{(n)}_{j} \cond M_1,\ldots, M_{k-1}, X_k, Y^{(n+1)},\ldots, Y^{(m)} \right]\ge 1$,
		the conditional distribution 
		$G'_j\eqdef X^{(n)}_{j} \Cond \set{M_1,\ldots, M_{k-1}, X_k, Y^{(n+1)},\ldots, Y^{(m)}}$ is a convex combination of distributions uniform over two values.
		More specifically, $G'_j=\sum_{v_j} \alpha_{v_j}\cdot G^{[v_j]}$, where $\alpha_{v_j}\in (0,1)$  and each $G^{[v_j]}$ is a random variable uniform over two values. % (say, $v_j^1$ and $v_j^2$).
		For $j\notin J$, $G'_j=\sum_{v_j} \alpha_{v_j} \cdot G^{[v_j]} $ where $G^{[v_j]}$ is fixed, i.e., a random variable that equals one value with probability $1$.
		For $v=(v_1,\dots,v_{k-1})$, let $\alpha_v=\prod_{j=1}^{k-1}\alpha_{v_j}$ and $G^{[v]}=\left(G^{[v_1]},\ldots,G^{[v_{k-1}]}, -\sum_{j=1}^{k-1} G^{[v_j]}\right)$,
		then $G'$ can be decomposed as
		$G'=\sum_{v} \alpha_v\cdot G^{[v]}$. 
		% and thus  
		% $G'_k=-\sum_{j=1}^{k-1} G'_j$.

		% (note that when realizing $G'_j$'s as random variables, we have $G'_k=\sum_{j=1}^{k-1} G'_j$ and hence with probability $\alpha_v$, $G'_k$ follows $G^{[v]}=\sum_{j=1}^{k-1} G^{[v_j]}$ when $G'_j=G^{[v_j]}$ for every $j\in [k-1]$). 
		% % Here we abuse notation a little and use the same $G'$ both for the random variable and the distribution. 

		Now for every $j\in J$ and $G^{[v_j]}$ uniform over $\set{a_j,b_j}\subset [a]$,
		we can assume w.l.o.g., $a_j<b_j$ and write $G^{[v_j]}=a_j + (b_j-a_j) Z_j$ where $Z_j$ is uniform over $\zo$.
		Since $b_j-a_j\in [a]$, among the $>k/2$ indices $j\in J$ for which $G^{[v_j]}$ takes two values, 
		we must have $t\ge |J|/a>k/\bigo{\log k}>\sqrt{k}$ indices $J'$ such that for any $j\in J'$ the value $b_j-a_j$ is the same value $M'$.

		Thus $G^{[v]}$ can be further decomposed into a convex combination of $G^\set{u}$ where, among the indices in $J$, only those in $J'$ are not fixed.
		Fix any $u$ and denote $G^\set{u}$ by $G''$.
		Let $S=\sum_{j\in J'} Z_j$ denote the sum of $t$ uniform i.i.d. $0/1$ random variables.
		Then we can write
		\begin{align*}
			G''_k &= b + M' S\\
			B''_k &= b + M' S + M
		\end{align*}

		Since $1\le M'<a$, $M'$ divides $M$ and hence $M=M'q$ for $q\in \Z$ and $q\le M\le k^{1/8}$.
		Now we can apply $q$ times the shift-invariance of the binomial distribution, which is stated as follows:

		\begin{claim}[Claim 39 in \cite{Viola:2013we}]
			Let $S$ be the sum of $t$ uniform, i.i.d. Boolean random variables.
			Then $S$ and $S+1$ have statistical distance $\le \bigo{1/\sqrt{\ell}}$. 
		\end{claim}

		This yields that $G''_k$ and $B''_k$ have statistical distance 
		$$SD(G''_k, B''_k)=SD(M' \cdot S, M'\cdot(q+S))\le q\cdot \bigo{1/\sqrt{\sqrt{k}}}\le k^{1/8}/k^{1/4}=k^{-1/8}$$

		Recalling that $G'$ is just a convex combination of $G''$, 
		the statistical distance between $G'_k$ and $B'_k=G'_k+M$ is also bounded by $k^{-1/8}$. 
		However, by definition of $G'_k$ and $B'_k$ we conclude that the referee cannot distinguish the two cases of $X^{(n)}\follows G$ and $X^{(n)}\follows B$ with advantage greater than $k^{-1/8}< 1/6$, which contradicts the condition that $\Pi'$ has error probability $<1/3$.
		
		Therefore, 
		$I\left(X^{(n)}_{-k}; M_1,\ldots, M_{k-1}\right)\ge 0.1 k\log a=\bigomega{k\log a}$.
	\end{proof}	

\bibliographystyle{tocplain}   %%% \bibliographystyle{plain}

%%% !!! AUTHOR
%%% Change this to match the name of your BIB file
\bibliography{biblio}

%%% !!! AUTHOR
%%% Include a short description of each author's affiliation
%%% following the structure below. Use the same unique ID used
%%% previously (presumably lower case last name).
%%% Use \tocat{} and \tocdot{} instead of "@" and "." in emails
\begin{tocauthors}
\begin{tocinfo}[du]
 Elbert Du\\
 4th year Undergraduate\\
 Department of Computer Science and\\
 Department of Mathematics\\
 Harvard University\\
 Cambridge, MA, USA\\
 edu@college.harvard.edu \\   %% email address here
\end{tocinfo}
\begin{tocinfo}[mitzenmacher]
 Michael Mitzenmacher\\
 Professor of Computer Science\\
 Department of Computer Science\\
 Harvard University\\
 Cambridge, MA, USA\\
 michaelm@eecs.harvard.edu \\   %% email address here
 \url{https://www.eecs.harvard.edu/~michaelm/}      %% your home page here
\end{tocinfo}
\begin{tocinfo}[woodruff]
 David Woodruff\\
 Associate Professor of Computer Science\\
 Department of Computer Science\\
 Carnegie Mellon University\\
 Pittsburgh, PA, USA\\
 dwoodruf@andrew.cmu.edu \\   %% email address here
 \url{http://www.cs.cmu.edu/~dwoodruf/}      %% your home page here
\end{tocinfo}
\begin{tocinfo}[yang]
 Guang Yang\\
 Research Director\\
 Tree-Graph Blockchain Innovation Center of Shanghai and\\
 Tree-Graph Blockchain Innovation Center of Xiang River Hunan\\
 Conflux Foundation\\
 Shanghai, China\\
 guang.research@gmail.com \\
 \url{https://sites.google.com/site/guangyangresearch/home}      %% your home page here
\end{tocinfo}
\end{tocauthors}

%%% !!!!
%%% Add a biographic sketch about each author.  The bio sections of
%%% previously published papers also appear in HTML, for example,
%%%   http://www.theoryofcomputing.org/articles/v003a009/about.html ).
%%% Include basic info about your education, research, and career
%%% (institutions, subject/title of dissertation, name of advisor, 
%%% list of areas of interest in some detail ["complexity theory" won't
%%% distinguish you from the majority of authors]).  Please also
%%% include information not readily available from other sources,
%%% like where you grew up, how you were first exposed to mathematics.
%%% The bio sketch is especially a good place to recognize an early
%%% mentor who helped turn your interest toward mathematics.  We also
%%% encourage you to include some more personal information (family,
%%% hobby, etc.).  Sprinkle it with humor.  IMPORTANT: Please include
%%% links in the bio sketch (to your advisor's and your Alma Mater's
%%% home page, your favorite hobby site, etc).  As far as links go,
%%% the more, the merrier. Use the following syntax:
%%%    Her advisor was \href{http://url}{name}.
%%% IMPORTANT: please remember to use ascii TeX codes for characters
%%% with a foreign accent.
\begin{tocaboutauthors}
\begin{tocabout}[surname]  %% use the same 
  \textsc{Elbert Du} is currently a senior at Harvard College, concentrating in mathematics and getting a fourth year masters with the AB/SM program in Computer Science. He was first introduced to the world of academic mathematics when he was in fifth grade, and he began attending the late Professor Paul Sally Jr.’s Young Scholars’ Program at the University of Chicago. Elbert is now interested in studying complexity, differential privacy, and adaptive data analysis. In his spare time, Elbert enjoys reading, solving chess puzzles, and playing video games.
\end{tocabout}
\begin{tocabout}[Mitzenmacher]
  \textsc{Michael Mitzenmacher} is a Professor of Computer Science at Harvard University.  He is the co-author of a well-known textbook on randomized algorithms and probabilistic
techniques in computer science with Eli Upfal.  He is an ACM and IEEE Fellow.  
\end{tocabout}
\begin{tocabout}[Woodruff]
\textsc{David Woodruff} is an associate
professor in the Computer Science Department at Carnegie Mellon University. He works on the foundations of data science, specifically in data streams, machine learning, randomized numerical linear algebra, sketching and sparse recovery. 
\end{tocabout}
\begin{tocabout}[Yang]
\textsc{Guang Yang} is currently the research director at Conflux, a startup blockchain project initiated by Fan Long and Andrew Yao. Before joining Conflux, he was an assistant professor at Institute of Computing Technology (ICT), Chinese Academy of Sciences. 
\end{tocabout}
\end{tocaboutauthors}

\end{document}